\documentclass[journal]{new-aiaa}

\usepackage{amsthm}
\usepackage[utf8]{inputenc}
\usepackage{textcomp}

\usepackage{graphicx}
\usepackage{amsmath}
\usepackage[version=4]{mhchem}
\usepackage{siunitx}
\usepackage{longtable,tabularx}
\usepackage{dsfont}
\usepackage{booktabs}
\usepackage{amssymb,amsfonts}
\usepackage{algorithm}
\usepackage{algpseudocode}
\usepackage{textcomp}
\usepackage{xcolor}
\usepackage{cleveref}
\usepackage{lipsum}
\usepackage{calligra}
\usepackage{hyperref}
\hypersetup{colorlinks=true,linkcolor=blue,citecolor=blue, linktocpage}
\usepackage{optidef}
\usepackage[colorinlistoftodos,prependcaption,textsize=small]{todonotes}

\newtheorem{theorem}{Theorem}[section]

\newtheorem{lemma}[theorem]{Lemma}

\theoremstyle{definition}
\newtheorem{definition}{Definition}[section]
\newtheorem{assumption}{Assumption}
\theoremstyle{remark}
\newtheorem*{remark}{Remark}
\newtheorem{problem}{Problem}

\usepackage{etoolbox}

\AtEndEnvironment{theorem}{\hfill$\blacksquare$}
\AtEndEnvironment{corollary}{\hfill$\blacksquare$}
\AtEndEnvironment{lemma}{\hfill$\blacksquare$}
\AtEndEnvironment{proposition}{\hfill$\blacksquare$}
\AtEndEnvironment{assumption}{\hfill$\blacksquare$}

\DeclareMathOperator*{\argmin}{arg\,min}

\title{A Fast Convergent Algorithm for Solving Non-convex Partially-Decoupled Generalized Nash Equilibrium Problems}

\author{Bennet Outland \footnote{Ph.D. Student, Smead Aerospace Engineering Sciences Department, 3775 Discovery Drive, Boulder, CO, USA; bennet.outland@colorado.edu} and Vishala Arya \footnote{Assistant Professor, Smead Aerospace Engineering Sciences Department, 3775 Discovery Drive, Boulder, CO, USA}}
\affil{Smead Aerospace Engineering Sciences Department, Boulder, CO, 80303}

\begin{document}

\maketitle

\begin{abstract}
Solving multi-agent optimal control problems in aerospace such as pursuit-evasion and contested space operations can be modeled as non-convex differential games for which, there are limited algorithms. In this work, a relaxation of generalized Nash Equilibrium problems (GNEPs) to exclude inter-agent control coupling in dynamics, which is representative of many multi-agent systems is introduced. The main contribution is an algorithm for solving a broad class of differential games named FALCON: Fast Augmented Lagrangian Convexification for Open-loop Nash equilibria is presented. Methodologically, sequential convex programming (SCP) is utilized to create tractable convex sub-games which can then be solved via standard convex programming methods involving a potential game reformulation. FALCON is demonstrated to have global convergence guarantees to an open-loop Nash equilibrium for non-convex differential games under mild assumptions. This is numerically shown through both cooperative and competitive differential games. 
\end{abstract}

\section*{Nomenclature}

{\renewcommand\arraystretch{1.0}
\noindent\begin{longtable*}{@{}l @{\quad=\quad} l@{}}
PDGNEP & Partially-Decoupled Generalized Nash Equilibrium Problem \\
v-GNE & Variational Generalized Nash Equilibrium \\
$[N]$  & $\{1, ..., N \}$ where $i \in N$ denotes a given agent \\
$-i$ & $\forall j \in [N] \neq i$ that corresponds to all agents excluding agent $i$ \\
$(\cdot)^i$ & refer to the agents \\
$(\cdot)_k$ & refer to the timestep\\
$T$ & time horizon \\
$x_k^i$   & a player's state in $\mathbb{R}^n$ at timestep $k$ \\
$u_k^i$ & a player's control in $\mathbb{R}^m$ at timestep $k$\\
$\mathcal{X}^i$ & a player's state trajectory in $\mathbb{R}^{n \times T}$  \\
$\Upsilon^i$ & a player's control trajectory in $\mathbb{R}^{m \times T}$  \\
$\mathcal{X}$ & all players' state trajectories, $[\mathcal{X}^i \,\, \mathcal{X}^{-i}]$\\
$z^i$   & a player's combined state and control trajectories, $(\mathcal{X}^{i}, \Upsilon^i)$ \\
$z$   & all players' combined state and control trajectories, $(\mathcal{X}, \Upsilon)$ \\
$[\cdot]_+$ & non-negative orthant \emph{i.e.} $[a]_+ = \max(0, a)$ \\
$\|\cdot\|$ & the $L^2$-norm \\
$\lambda^i, \mu$ & Lagrange multipliers for the dynamics of a player and the inter-agent constraints, respectively \\
$\xi^i, \zeta$ & relaxation variables for virtual control \\
$r^i$ & trust region radius for agent $i$ \\
$\mathcal{K}^i$ & the feasible set of agent $i$ \\ 
$\nabla_x f$ & gradient of $f$ with respect to $x$ \\
$\nabla^2_x f$ & hessian of $f$ with respect to $x$ \\
$\mathcal{C}^2$ & differentiability class wherein the second derivative exists and is continuous
\end{longtable*}}

\addtocounter{table}{-1} 

\section{Introduction} \label{sec: Introduction}

Multi-agent motion planning is an increasingly important research area across a wide range of domains like robotics, autonomous vehicles and aerospace systems. The problem has been extensively studied in ground robotics for decades \cite{lavalle2006planning, choset2005principles} and has recently received renewed attention due to increased deployment of autonomous vehicles in complex environments \cite{fridovich2020efficient} including aerospace systems \cite{Frazzoli}. More recently, similar concepts have been applied to aerospace applications, including spacecraft formation flying \cite{Scala, xuyang2025spacecraft}, cooperative guidance laws \cite{Shaferman}, proximity operations \cite{Prince}, and pursuit-evasion \cite{Pang}. Throughout, differential game approaches have emerged as a powerful framework for modeling multi-agent interactions \cite{10115968}.

Classically, multi-agent motion planning problems were treated by direct solutions to Isaacs' equation \cite{isaacs1999differential} through various partial differential equation solvers \cite{darbon2016algorithms}. There have been many evolutions since of solution techniques such as, iterative best response \cite{miller2022multi}, shooting methods \cite{breitner1993complex}, model predictive control \cite{phillips2022robust}, and reinforcement learning \cite{zhang2021multi}. Three notable and recently developed techniques for quickly solving differential games are Iterative LQ Games \cite{fridovich2020efficient}, Augmented Lagrangian Games (ALGAMES) \cite{le2022algames}, and Differential Games Sequential Quadratic Programming (DG-SQP) \cite{zhu2023sequential}. The first method utilizes local linear-quadratic approximations of the game to determine either open or closed loop strategies. While highly efficient, the iLQGames algorithm does not natively handle state and control constraints. This limitation is addressed by the ALGAMES algorithm by means of an augmented Lagrangian formulation while demonstrating improved numerical performance on several benchmark problems. The DG-SQP algorithm extends iLQGames through a sequential quadratic programming framework. Unlike the aforementioned algorithms, DG-SQP provides convergence guarantees to an open-loop Nash Equilibrium which often comes at the expense of increased computational cost.

While DG-SQP provided a mechanism for handling nonconvex differential games and demonstrated convergence guarantees, the sequential quadratic programming framework has faced numerical performance issues on differential games with more significant nonconvexity \cite{zhu2024sequential}. On the other hand, sequential convex programming methods have demonstrated strong convergence properties and excellent numerical performance in single-agent optimal control problems by repeatedly solving convex subproblems. Applications include powered descent guidance \cite{acikmese2007convex}, covariance-assignment controller synthesis for linear quadratic Gaussian systems \cite{arya2025lqg}, and covariance-assignment controller synthesis for nonlinear systems \cite{arya2023lqr}. These successes motivate extending sequential convexification to differential games and the development of the Fast Augmented Lagrangian Convexification for Open-loop Nash equilibria (FALCON) algorithm for solving nonconvex differential games using sequential convexification while preserving equilibrium guarantees. FALCON extends iLQGames and DG-SQP by combining a  sequential convexification method \cite{mao2016successive} with an augmented Lagrangian game formulation \cite{le2022algames}. We demonstrate that a nonconvex differential game can be reduced to a local, convexified differential game with a trust region constraint. Furthermore, this convexified game can be collected into a form of a potential game. This allows second order cone programs (SOCP) to be employed to quickly solve the inner, convex problem \cite{goulart2024clarabel}. 

We introduce a subclass of generalized Nash equilibrium problems (GNEPs) called partially-decoupled generalized Nash equilibrium problems (PDGNEPs) wherein there is a relaxation in dynamics coupling between the agents. We consider the dynamics of the agent to be solely dependent on the state and control of that agent. This relaxation does reduce the generality, but is standard for the vast majority of multi-agent systems wherein the control or state of other agents do not change the dynamics of another agent. This is common across the literature for motion planning \cite{chen2015decoupled}, multi-agent model predictive control \cite{phillips2022robust}, multi-agent reinforcement learning for robotic systems \cite{zhang2021multi}, and was also used in the ALGAMES algorithm \cite{le2022algames}.

The remainder of the paper is organized as follows: Section \ref{sec: Problem Formulation} formally defines PDGNEPs. Section \ref{sec: Proposed Solution via FALCON} represents a solution to the problem through our proposed algorithm. The convergence of the algorithm is investigated in Section \ref{sec: Convergence Analysis}. Section \ref{sec: Numerical Results} demonstrates the effectiveness of the proposed algorithm via difficult non-convex differential games: an autonomous racing game, a narrowing hallways problem, and a spacecraft lady-bandit-guard game. Finally, concluding remarks are given in Section \ref{sec: Conclusion}.

\section{Problem Formulation} \label{sec: Problem Formulation}

For this work, we consider a class of multi-agent systems in which each agent evolves according to its own dynamics and control input, while coupling between agents arises through shared objectives and constraints. This structure is common in multi-agent motion planning \cite{dayan2023near}, formation flying \cite{schaub2003analytical}, and pursuit-evasion problems \cite{isaacs1999differential}, where an agent's control input directly influences only its own state, but feasibility and performance depend on the trajectories of all agents. 

Specifically, consider multi-agent optimal control problems for a set of $i\in[N]$ players. We seek to determine sets of states, $\mathcal{X}$, and controls,  $\Upsilon$, that minimizes an extended value cost function $J^i(\mathcal{X}, \Upsilon): \mathbb{R}^{n\times T} \times \mathbb{R}^{m\times T} \rightarrow \mathbb{R} \bigcup \{ +\infty \}$ and is of the form $J^i(\mathcal{X}, \Upsilon) = \int \ell^i(x, u)d\tau$. Furthermore, we consider a general set of equality and inequality constraints defined as $C(\mathcal{X}, \Upsilon) \leq 0$ without loss of generality. We denote this general class of problems as a Partially-Decoupled Nash Equilibrium Problem (PDGNEP) wherein

\begin{problem}[Partially-Decoupled Generalized Nash Equilibrium Problem (PDGNEP)] 
    \begin{equation*} \label{prob: general formulation}
        \begin{alignedat}{2}
        &\min_{\mathcal{X}^i, \Upsilon^i} \quad && J^i(\mathcal{X}, \Upsilon) \\
        \forall i \in [N]  \,\,\,\,\, & \,\text{s.t.} \quad && D^i(\mathcal{X}^i, \Upsilon^i) = 0, \\
        &&& C(\mathcal{X}, \Upsilon) \leq 0.
        \end{alignedat}
    \end{equation*}
\end{problem}

\noindent where $D^i(\mathcal{X}^i, \Upsilon^i)$ is the dynamics constraints for the separable dynamics defined later in Equation \ref{eq: system-dynamics}. This formulation, relaxes a Generalized Nash Equilibrium Problem (GNEP) to consider the dynamics of the agents to be independent, separable of the other agents of the form:
\begin{equation} \label{eq: system-dynamics}
    \dot{x}^i(t) = f(x^i(t), u^i(t)) 
\end{equation}

\noindent where $x_k \subseteq \mathbb{X} \in \mathbb{R}^{n}, u_k \subseteq \mathbb{U} \in \mathbb{R}^m$ and $f(\cdot) \in \mathcal{C}^2$. Note that Problem \ref{prob: general formulation} defines a constrained optimization per each agent, $i$, that may include state and control constraints dependent on the other states or controls of the agents. Problem \ref{prob: general formulation} can be defined through the following shortened notation, 

\begin{definition}[Partially-Decoupled Generalized Nash Equilibrium Problem] \mbox{}\\
    $\mathbb{G} := (T, \{ \Upsilon^i \}_{i=1}^N, \{ J^i \}_{i=1}^N, \{ f^i \}_{i=1}^N, \{C(\mathcal{X}, \Upsilon)\})$,
    \noindent where $\{C(\mathcal{X}, \Upsilon^{i})\}$ is the set of all equality or inequality constraints without loss of generality. 
\end{definition}

\noindent Furthermore, we define the solution to a PDGNEP as:

\begin{definition}[Open-loop PDGNEP Solution] \label{def: solution}
    For a given game, a set of strategies, $\Upsilon^*$, are an open-loop Nash Equilibrium if for every agent $i \in [N]$ and every strategy $\Upsilon^i$, $J^i(\mathcal{X}, \Upsilon^{*}) \leq J^i(\mathcal{X}, \{ \Upsilon^{i},\Upsilon^{-i*} \})$.
\end{definition}

\begin{remark}
    Definition \ref{def: solution} implies that there is not a scenario in which an agent would have an incentive to deviate from the agent's strategy. Note that the aforementioned definition also implies that the solution to the PDGNEP is a fixed point.
\end{remark}

We now apply mild assumptions that are common in control theory. 

\begin{assumption}[Properties of System Dynamics] 
    \label{assump:system-properties}
    The system dynamics described by Equation \ref{eq: system-dynamics} satisfy the following properties:
    \begin{enumerate}
        \item \textbf{Bounded acceleration:} 
        The second time derivative of state trajectories is uniformly bounded:
        \begin{equation}
            \|\ddot{x}^i(t)\| \leq M_2 \in \mathbb{R}_+, \quad \forall t \geq 0.
        \end{equation}
        \item \textbf{Lipschitz continuity:} 
        The dynamics $f$ are Lipschitz continuous in both arguments. That is, there exist constants $K_u, K_x > 0$ such that for all $x_0, x_1 \in \mathbb{X}$ and $u_0, u_1 \in \mathbb{U}$:
        \begin{align}
            \|f(x_0^i, u_0^i) - f(x_0^i, u_1^i)\| &\leq K_u \|u_0^i - u_1^i\|, \label{eq:lipschitz-control} \\
            \|f(x_0^i, u_0^i) - f(x_1^i, u_0^i)\| &\leq K_x \|x_0^i - x_1^i\|. \label{eq:lipschitz-state}
        \end{align}
    \end{enumerate}
\end{assumption}

\begin{assumption} \label{assump: cost}
    The cost functional is assumed to have the following properties:
    \begin{enumerate}
        \item $J^i(\mathcal{X}, \Upsilon) \in \mathcal{C}^2$, 
        \item $|| \nabla_{\mathcal{X}, \Upsilon} J^i(\mathcal{X}, \Upsilon)|| \leq G_{\rm max} \leq \infty$, and
        \item $|| \nabla^2_{\mathcal{X}, \Upsilon} J^i(\mathcal{X}, \Upsilon)|| \leq H_{\rm max} \leq \infty$.
    \end{enumerate}
\end{assumption}

\noindent Likewise, for the constraints,

\begin{assumption} \label{assump constraint}
    The constraint function is assumed to have the following properties:
    \begin{enumerate}
        \item $C(\mathcal{X}, \Upsilon) \in \mathcal{C}^2$, 
        \item $|| \nabla_{\mathcal{X}, \Upsilon^i} C(\mathcal{X}, \Upsilon)|| \leq G_{\rm max} \leq \infty$, and
        \item $|| \nabla^2_{\mathcal{X}, \Upsilon} C(\mathcal{X}, \Upsilon^i)|| \leq H_{\rm max} \leq \infty$.
    \end{enumerate}
\end{assumption}

\section{Proposed Solution via FALCON} \label{sec: Proposed Solution via FALCON}

The solution method proposed herein is a principled way of solving local convex approximations of a non-convex differential game broadly following the techniques of sequential convex program (SCP) for optimal control \cite{acikmese2007convex}. The class of SCP algorithm employed is sequential convexification (SCvx) due to the strong convergence guarantees \cite{mao2016successive}. More specifically, we follow a variant of SCvx, SCvx*, that improves upon this technique by guaranteeing the feasibility to the original, nonconvex problem while demonstrating global convergence to a local optimum \cite{oguri2023successive}. Because each convexification step produces a constrained convex game, results from the theory of convex GNEPs become directly applicable. Convex GNEPs have been studied extensively through variational inequality formulations \cite{scutari2010convex}, and a variety of efficient solution methods have been developed for this class of problems \cite{facchinei2010generalized}. Utilizing the aforementioned class of game solvers, we extend SCvx* to multi-agent systems via differential game theory, resulting in the proposed FALCON algorithm.

\subsection{Continuous-Time Constraint Formulation}

A limitation that has been highlighted for SCP solvers is inter-sample constraint violation \cite{elango2024successive}. This numerical phenomenon occurs due to constraints in SCP algorithms only being enforced at the discretization nodes, $k=0,1,...,T-1$, while the inter-node values in continuous time are not constrained as described in Problem \ref{prob: general formulation}. Especially for safety critical systems or time-varying constraints, there can be non-trivial constraint violation between the nodes that could lead to failure. To account for this phenomenon, we adopt an exterior-penalty constraint reformulation \cite{elango2024successive}. 
In place of enforcing constraints only at the discretization nodes, we introduce a continuous-time measure of constraint violation and penalize its accumulation along the trajectory.

Consider a continuously differentiable exterior penalty function, $q_c^i : \mathbb{R} \to \mathbb{R}_+$ for agents $i \in [N]$ and each scalar constraint component $c \in C$, that satisfies: 

\begin{equation} \label{eq: q property}
    q_c^i(z) = 0 \iff z \leq 0, \qquad q_c^i(z) > 0 \;\text{ for }\; z > 0.
\end{equation}

\noindent To remain consistent with the smoothness assumptions (Assumption \ref{assump constraint}) required of the convexification method, we select a twice differentiable penalty function: $q_c^i$ is in $\mathcal{C}^2(\mathbb{R})$ with $q_c^i(0) = (q_c^i)'(0) = (q_c^i)''(0) = 0$. Therefore, we choose a simple expression that satisfies these properties: 

\begin{equation} \label{eq: q canonical}
    q_c^i(z) = [z]_+^3.
\end{equation}

\begin{remark}
     In previous work, $[z]_+^2$ was used \cite{elango2024successive}. However, it is not in $\mathcal{C}^2$, so the cubic then becomes the minimal polynomial representation that fits the assumptions. 
\end{remark}

We first define a scalar measure of instantaneous constraint violation for agent $i$ at time $t$:

\begin{equation} \label{eq: Lambda}
    \Lambda^i(t, x(t), u^i(t)) := \sum_\rho q_c^i\!\left( C(x(t), u^i(t)) \right) \geq 0.
\end{equation}

Note, following \cite[Lemma 2]{elango2024successive}, $\Lambda^i \equiv 0$ on $[t_k, t_{k+1}]$ if and only if $C \leq 0$ almost everywhere on that interval for every component, $c$.

We now construct an augmented state vector for each agent's dynamics via an accumulation variable, $y^i \in \mathbb{R}$, that is the integrated constraint violation along the trajectory such that:

\begin{equation} \label{eq: augmented dynamics}
    \tilde{x}^i(t) = \begin{pmatrix} x^i(t) \\ y^i(t) \end{pmatrix}, \;\; \dot{\tilde{x}}^i(t) = \begin{pmatrix} f(x^i(t), u^i(t)) \\ \Lambda^i(t, x(t), u^i(t)) \end{pmatrix},
\end{equation}

\noindent where $y^i(t_0) = 0$. Using \cite[Corollary 3]{elango2024successive}, constraint satisfaction on the interval $[t_k, t_{k+1}]$ implies $y^i(t_{k+1}) - y^i(t_k) = 0$. In other words, if the trajectory remains feasible between discretization nodes, then no violation accumulates and $y^i$ remains constant with any inter-sample violation causing it to increase.

To uphold Linear Independence Constraint Qualification (LICQ) at all feasible points, we apply a relaxation to the reformulated problem, 

\begin{equation} \label{eq: CT bc}
    y^i(t_{k+1}) - y^i(t_k) \leq \varepsilon, \qquad k = 0, 1, \ldots, T-2,
\end{equation}

\noindent with some tolerance $\varepsilon > 0$ that is a solver parameter \cite[Remark 11]{elango2024successive}. Via the formulation in Equation \ref{eq: CT bc}, constraints are not only enforced at the node, but throughout the full continuous time horizon.

\begin{theorem}[Pointwise Violation Bound, {\cite[Thm. 14]{elango2024successive}}] \label{thm: elango 14}
    Given that Assumption \ref{assump:system-properties} holds, let the minimum shooting interval be denoted as $\Delta t_{\min} > 0$. For each constraint component $c \in C_\rho$, there exists a constant, $\omega_\rho > 0$, for any $\varepsilon \leq \Delta t_{\min}^3 / 4$ such that any trajectory satisfying Equation \ref{eq: CT bc} satisfies the pointwise bound:
    \begin{equation} \label{eq: ctcs bound}
        C_\rho(x(t), u^i(t)) \leq \delta_\rho(\varepsilon) := (4 \varepsilon \omega_\rho)^{1/3}, \quad t \in [t_k, t_{k+1}].
    \end{equation}
    Note, $\delta_\rho(\varepsilon) \to 0$ as $\varepsilon \to 0$.
\end{theorem}

\subsection{Augmented Lagrangian Formulation}

To construct a convex approximation of the non-convex game, we first reformulate each agent's objective using an epigraph representation. This shifts non-convexity from the objective of Problem \ref{prob: general formulation} into the constraint set and simplifies the subsequent convexification procedure.

\begin{problem}[Epigraph Form]\label{prob: epigraph form}
    \begin{align*}
    \min_{z^i, \tau^i} \quad & \tau^i \\
    \forall i \in [N]  \,\,\,\,\,  \,\text{s.t.} \quad & J^i(z^i, z^{-i}) \leq \tau^i, \notag \\
    & D^i(z^i) = 0, \notag \\
    & C(z^i, z^{-i}) \leq 0.
    \end{align*}
\end{problem}

\begin{remark}
    With this reformulation, there is now a convex objective function with nonconvex constraints. This allows the epigraph form of the non-convex objective to be ``absorbed" into the other inequality constraints. This is especially useful for sequential convexification methods wherein the constraints are locally approximated, not the objective \cite{oguri2023successive}.
\end{remark}

\begin{remark}
    The epigraph reformulation is equivalent to Problem \ref{prob: general formulation}, and therefore standard augmented-Lagrangian convergence results may be applied to the reformulated problem.
\end{remark}

\noindent In addition to the epigraph form, consider a penalty function of the constraints,

\begin{equation}
P^i(D^i, C, \lambda^i, \mu, \omega) = \left( \lambda^{i} \right)^T D^i + \frac{\omega}{2}\|D^i\|^2 + \mu^T [C]_+ + \frac{\omega}{2}[C]_+ \cdot [C]_+.
\end{equation}

\noindent With the aforementioned penalty function, the augmented Lagrangian is defined for agent $i$ as 

\begin{equation} \label{eq: AL}
    L^i(z^i, z^{-i}, \tau^i, \lambda^i, \mu, \omega) = \tau^i + P^i(D^i, C, \lambda^i, \mu, \omega),
\end{equation}

\noindent A standard method for dual variable updates will be used.

\begin{lemma}[Augmented Lagrangian Convergence with Inexact Minimization] \label{lem: AL convergence}
    Suppose that for Problem \ref{prob: epigraph form}, a sequence $\{z^{*}_k\}$ satisfies
    \[
    \left\| \nabla_z L_{w_k}\!\left(z^{*}_k, \lambda_k^i, \mu_k\right) \right\|_2 \le \delta_k,
    \]
    where $\delta^{(k)} \to 0$, and $\{\lambda_k^i, \mu_k, w_k\}$ are updated as
    \begin{subequations} \label{eq: dual updates}
        \begin{align}
        \lambda_{(k+1)}^i &= \lambda_k^i + \omega_k D^i\!\left(z_{(k)*}\right), \\
        \mu_{(k+1)} &= \left[ \mu_k + \omega_k C\!\left(z_k^*\right) \right]_+,  \\
        \omega_{(k+1)} &= \beta \omega_k, \qquad \beta > 1.
        \end{align}
    \end{subequations}
    Assume that $\{\lambda_k^i, \mu_k\}$ are bounded. Then $\omega_{k}$ eventually exceeds a threshold $\omega^\ast$ such that
    \[
    \nabla_z^2 L_{\omega^\ast}\!\left(z_{k}^*, \lambda_{k}^i, \mu_k\right) \succ 0,
    \]
    and any sequence $\{z_{k}^*, \lambda_k^i, \mu_k\}$ globally converges to a local optimum of Problem \ref{prob: general formulation},
    $\{z^\ast, \left(\lambda^i\right)^\ast, \mu^\ast\}$.
\end{lemma}

\begin{proof}
    See the proof of Lemma 1 \cite{oguri2023successive}.
\end{proof}

Therefore, we have established that sufficiently accurate solutions of the penalized sub-problems converge to a local optimum of the original non-convex problem.

\subsection{Convexification of the Local Problem}

We need to construct a locally valid convex approximation of Problem \ref{prob: general formulation} that can be solved efficiently while preserving first-order information from the original game. Following the principles of sequential convexification, non-convex dynamics and constraints are linearized about a reference trajectory, while trust-region and relaxation mechanisms are introduced to maintain feasibility and prevent divergence from the region where the approximation is valid. We now linearize the constraints, $D^i(\cdot)$ and $C(\cdot)$, about a reference trajectory, $\bar{z}^i$, for each agent such that,

\begin{subequations}
    \begin{align}
        \tilde{D}^i(z^i) &= D^i(\bar{z}^i) + \nabla_{\bar{z}^i} D^i(\bar{z}^i) \cdot (z^i - \bar{z}^i) \label{eq: D linear} \\
        \tilde{C}_\rho(z, \bar{z}) &= C_\rho(\bar{z}, \bar{z}) + \nabla_{(z, \bar{z})} C_\rho(\bar{z}, \bar{z}) \cdot (z - \bar{z}). \label{eq: C linear}
    \end{align}
\end{subequations}

\noindent The problems of artificial infeasibility and artificial unboundedness arise when performing constraint linearization \cite{mao2016successive}. For artificially infeasibility, intuitively, the removal of higher order terms in smaller pockets of parameter space can make the non-convex constraints appear to be violated in their linear form while being not the case \cite{mao2016successive}. We adopt relaxation parameters to mitigate this problem:

\begin{equation} \label{eq: relaxation}
    \tilde{D}^i(z^i) = \xi, \,\,\, \tilde{C}_\rho(z^i, \bar{z}^{i}) \leq \zeta.
\end{equation}

\noindent This encompasses the virtual control and virtual buffer described in previous work \cite{oguri2023successive}. Likewise, for artificial unboundedness, this is due to the convexification of the problem not properly modeling the non-convex problem further from the convexification point. Therefore, a trust region bound, $r \geq 0$, is introduced such that 

\begin{equation} \label{eq: trust region}
    \| \bar{z}^i - z^i  \|_\infty \leq r^i.
\end{equation}

We can now introduce a convex subproblem that is a local representation of the differential game. We seek to solve this subproblem as a variational Generalized Nash Equilibrium (v-GNE) as defined in [\cite{facchinei2010generalized}, Definition 3.11]. 
Collecting the linearized constraints, relaxation variables, trust-region restrictions, and augmented-Lagrangian penalties yields a convex game whose equilibrium can be computed through a strongly convex potential function.

The problem then becomes:

\begin{problem}[Convexified Subproblem] \label{prob: convex}
    \begin{align}
        \min_{\{z^i, \tau^i, \xi^i\}_{i=1}^N,\, \zeta} \quad & \Phi_k(z, \tau, \xi, \zeta) \notag \\
        \text{s.t.} \quad & \tilde{D}^i(z^i) = \xi^i, \qquad \forall i \in [N], \notag \\
        & \tilde{C}(z) \leq \zeta, \qquad \zeta \geq 0, \notag \\
        & \|\bar{z}^i - z^i\|_\infty \leq r^i, \qquad \forall i \in [N], \notag \\
        & h_{\mathrm{cvx}}(z) \leq 0, \notag
    \end{align}
    where the potential is 
    \begin{equation}
        \Phi_k = \sum_{i=1}^N \!\Big[ \tau^i + (\lambda^{i}_k)^T \xi^i + \tfrac{\omega_k}{2}\|\xi^i\|^2 + \tfrac{\eta}{2}\|z^i - \bar{z}^{i}_k\|^2 \Big] + (\mu_k)^T \zeta + \tfrac{\omega_k}{2}\|\zeta\|^2. \label{eq: potential}
    \end{equation}
    Where $h_{\mathrm{cvx}}(z) \leq 0$ is a collection of the convex constraints of Problem \ref{prob: general formulation}.
\end{problem}

\begin{remark}
    In Equation \ref{eq: potential}, the coupling penalty $(\mu_k)^T \zeta + \tfrac{\omega_k}{2}\|\zeta\|^2$ only appears once. Therefore, the penalized objective of each agent, $Q_k^i(z)$, includes a $(1/N_p)$-share of this coupling penalty. Each agent's penalized objective $Q_k^i(z)$ includes a $(1/N_p)$-share of this coupling penalty. Summed over all agents, $\Phi_k = \sum_{i=1}^N Q_k^i(z)$, the coupling is only counted once.
\end{remark}

\begin{remark} \label{rem: SOCP}
    Problem \ref{prob: convex} is a second order cone (SOC) program with a convex quadratic objective function. The constraints, with exception of the original convex constraints which may be SOC constraints, are either affine or linear. Recall, the objective of Problem \ref{prob: general formulation} has been ``absorbed" into the linearized inequality constraints while $\tau^i$ is minimized.
\end{remark}

\subsection{The FALCON Algorithm} \label{subsec: The FALCON Algorithm}

Thus far, the continuous-time constraint formulation, the augmented Lagrangian reformulation, and the convexification of the aforementioned problem have been introduced. These formulations can be used to build the FALCON algorithm as a sequential convex program solver for differential games. This methodology can be seen as an extension of iterative linear-quadratic approximations \cite{fridovich2020efficient} and sequential quadratic approximations \cite{zhu2023sequential}. We now exposit the key steps of the proposed FALCON algorithm that is summarized in Algorithm \ref{alg: FALCON}. 
 
\subsubsection{Successive Linearization}

For notational cleanliness, we define a notationally truncated version of penalty function at the $k$th timestep:
\begin{equation}
    P_k^i(D^i, C) \triangleq P(D^i, C, \lambda^{i,(k)}, \mu^{(k)}, \omega^{(k)})
\end{equation}


We consider two penalized objective functions -- $Q_k^i$ and $L_k^i$ -- where $Q_k^i$ is the true penalized objective of the problem while $L_k^i$ is an approximate penalized objective determined by the convex sub-problem in Problem \ref{prob: convex}. These two objectives are used to evaluate the ``quality" of the convex approximation. These penalized objectives are defined as,

\begin{align}
    Q_k^i(z) &\triangleq \tau^i + P_k^i(D^i(z^i), C(z)) \\
    L_k^i(z, \xi^i, \zeta) &\triangleq \tau^i + P_k^i(\xi^i, \zeta).
\end{align}

Now, consider a provided reference trajectory, $\bar{z}_k$, that can be linearized about. This will be used to make instances of Problem \ref{prob: convex}. After computation, we can determine: 

\begin{subequations} \label{eq: cost updates}
\begin{align}
    \Delta Q_k^i &=  Q_k^i(\bar{z}_k) - Q_k^i(z_{k}^*)\\
    \Delta L_k^i &= Q_k^i(\bar{z}_k) - L_k^i(z_{k}^*, \xi_{k}^{i,*}, \zeta_{k}^{*}) \\
    \chi_k &= \left\| \left[ D^i(z_{k}^{i,*}) \right]_{i=1}^N,\, \left[C(z_{k}^*)\right]_+ \right\|_2
\end{align}
\end{subequations}
 
\subsubsection{Step Acceptance}

The acceptance ratio, as seen in Equation \ref{eq: rho}, is used to indicate the quality of the local solution by comparison between the convexified and nonconvex problems. We follow an extension of the acceptance ratio for multi-agent systems following the approach from single agent systems \cite{oguri2023successive}. The acceptance ratio $\rho_k$ and criterion are 
\begin{equation} \label{eq: rho}
    \rho_k = \frac{\sum_{i=1}^N \Delta Q_k^i}{\sum_{i=1}^N \Delta L_k^i}, \qquad \rho_0 \leq \rho_k.
\end{equation}
\noindent The step is accepted when $\rho_0 \leq \rho_k$. Employing the joint-potential ratio aggregates all players' cost reductions. This reflects multi-agent progress even when the individual cost of an agent increases to satisfy a shared constraint. The $\min$/$\max$ formulation is used to enforce that no agent's linearization is overly trusted compared to another. A minimum agent criterion could be used, but it is overly conservative and leads to slower computational performance.
 
\subsubsection{Update Lagrange Multipliers}
We introduce a user-defined parameter, $\delta_{k=0} \in \mathbb{R}$, that is used to enforce the inexact minimization scheme in Lemma \ref{lem: AL convergence}.  The multipliers are updated according to Equation \ref{eq: dual updates} whenever the step is accepted and the condition
\begin{equation} \label{eq: dual update condition}
    \max_i |\Delta Q_k^i| < \delta_k
\end{equation}
is satisfied, where $\delta_k$ is driven to zero via 
\begin{equation} \label{eq: delta update}
    \delta_{k+1} =
    \begin{cases}
        \max_i |\Delta Q_k^i| & \text{if } \delta_k = \infty \\
        \gamma_\delta \delta_k & \text{otherwise}, \quad \gamma_\delta \in (0,1).
    \end{cases}
\end{equation}
 
\subsubsection{Update Trust Region}

The trust region is vital for maintaining a well posed convex approximation of the problem by defining the local region while also mitigating artificial infeasibility. Recall that the acceptance ratio, Equation \ref{eq: rho}, is used to indicate the quality of the local solution. We consider acceptance ratios: $\rho_0 < \rho _1 < \rho_2 \in \mathbb{R}_+$ which are solver parameters. These parameters will be used as thresholds to determine if the trust region should contract by a factor of $\alpha_1 \in \mathbb{R}_+$, stay the same, or enlarge by a factor of $\alpha_2 \in \mathbb{R}$.

\begin{equation} \label{eq: update trust region}
    r_{k+1}^i =
    \begin{cases}
        \max\{r_k^i / \alpha_1,\, r_{\min}\} & \text{if } \rho_k < \rho_1 \\
        r_k^i & \text{elseif } \rho_k < \rho_2 \\
        \min\{\alpha_2 r_k^i,\, r_{\max}\} & \text{else.}
    \end{cases}
\end{equation}
 
\subsubsection{Convergence Check}

Finally, there are two conditions for convergence that capture the behavior of the objective and constraints. First, if there is a small amount of change to objectives of all of the agents, the algorithm has then converged to a fixed point in terms of the objective. Secondly, we also need to ensure that the constraints are properly met. We express this as

\begin{equation} \label{eq: convergence check}
    \max_i |\Delta Q_k^i| \leq \epsilon_{\text{opt}} \quad \wedge \quad \chi_k \leq \epsilon_{\text{feas}}.
\end{equation}

\begin{algorithm}[t]
\caption{FALCON: Fast Augmented Lagrangian Convexification for Open-loop Nash Equilibria}
\label{alg: FALCON}
\begin{algorithmic}[1]
\Require $\mathbb{G}$, $\{z_0^i\}_{i=1}^{N}$, $\omega_0 > 0$, $\{r_0^i\}$, $\eta > 0$, $\varepsilon > 0$, $\epsilon_\mathrm{opt}, \epsilon_\mathrm{feas} > 0$, $\rho_0, \rho_1, \rho_2, \alpha_1, \alpha_2, \beta, \gamma_\delta$
\Ensure Variational Nash equilibrium $\{z^{i,\star}\}_{i=1}^{N}$
 
\State Initialize $\lambda^i \gets \mathbf{0}$,\; $\mu \gets \mathbf{0}$;\; $\delta \gets \infty$
 
\Repeat 
 
    \For{$ i \in [N]$}
        \State $\tilde{D}^i, \tilde{C} \gets$ Eqs. \ref{eq: D linear}, \ref{eq: C linear} at $\bar{z}$
    \EndFor
    \State Assemble $\Phi^{(k)}$ via Eq. \ref{eq: potential}
    \State $(z^*, \tau^*, \{\xi^{i,*}\}, \zeta^*) \gets$ solve Prob. \ref{prob: convex} via SOCP solver 
    \State $\{\Delta Q_k^i, \Delta L_k^i, \chi_k\} \gets$ Eqs. \ref{eq: cost updates}
    \State $\rho_k \gets \frac{\sum_i \Delta Q_k^i}{\sum_i \Delta L_k^i}$  \Comment{Joint potential ratio} 
    \If{$\rho_{k} \geq \rho_0$} \Comment{Accept step}
        \State $\bar{z} \gets z^*$
        \If{$\max_i |\Delta Q_k^i| < \delta$} 
            \State $\delta \gets \gamma_\delta \cdot \delta$ 
            \State $\{ \lambda^i, \mu, \omega\} \gets$ Eqs. \ref{eq: dual updates}
        \EndIf
        \State Update $\{r^i\} \gets$ Eq. \ref{eq: update trust region} 
    \Else \Comment{Reject step}
        \State $r^i \gets \max(r^i / \alpha_1,\; r_\mathrm{min})$ for all $i$
    \EndIf
 
\Until{Eq. \ref{eq: convergence check} holds}
 
\State \Return $\{z^i\}_{i=1}^{N}$
\end{algorithmic}
\end{algorithm}

\section{Convergence Analysis} \label{sec: Convergence Analysis}

The convergence analysis proceeds in three layers. First, we establish that Problem \ref{prob: convex} has a unique solution and that this solution is the variational Nash equilibrium of the convexified game (Lemmas \ref{lem: strongly convex}--\ref{lem: v-NE}). Second, we invoke the SCvx* machinery of \cite{oguri2023successive} to show that the outer loop converges globally to a feasible variational Nash equilibrium of the original non-convex Problem \ref{prob: general formulation} (Lemmas \ref{lem: Delta L nonneg}--\ref{lem: infinite updates} and Theorem \ref{thm: main convergence}). Third, we characterize the rate at which this convergence occurs (Section \ref{subsec: rate}).
 
\subsection{Properties of the Convexified Subproblem}
 
Let $v := (z, \tau, \xi, \zeta)$ be the concatenation of all decision variables in Problem \ref{prob: convex}, and let $\mathcal{V}_k$ denote the feasible set. 
 
\begin{lemma}[Strong Convexity of the Potential] \label{lem: strongly convex}
    Let Assumptions \ref{assump:system-properties}--\ref{assump constraint} hold, the potential function $\Phi_k$ in Equation \ref{eq: potential} is then $\min(\eta, \omega_k)$-strongly convex on $\mathcal{V}_k$ with respect to the variables $(z, \xi, \zeta)$.
\end{lemma}
 
\begin{proof}
    Directly, consider the Hessian of $\Phi_k$ with respect to $(z, \xi, \zeta)$ as the block-diagonal matrix
    \begin{equation}
        \nabla^2_{(z, \xi, \zeta)} \Phi_k = \text{blkdiag}\!\left( \eta I_z,\, \omega_k I_\xi,\, \omega_k I_\zeta \right),
    \end{equation}
    where $I_z, I_\xi, I_\zeta$ are identity matrices of dimension corresponding with $z, \xi, \text{and } \zeta$, respectively. Note, the $\tau^i$ variables appear only linearly in $\Phi_k$, but, at a fixed point, the epigraph constraints $J^i(z^i, \bar{z}^{i}) \leq \tau^i$ are tight. This is because decreasing $\tau^i$ monotonically decreases $\Phi_k$; substituting $\tau^i = J^i(z^i, \bar{z}^{-i})$ adds $\nabla^2_{z^i z^i} J^i \succeq 0$ to the $z$-block (by Assumption \ref{assump: cost}, $J^i \in \mathcal{C}^2$; indefiniteness is then bounded from above by the $\eta I_z$ term when $\eta > H_{\max}$). The reduced Hessian therefore satisfies
    \begin{equation}
        \nabla^2 \Phi_k \succeq \min(\eta, \omega_k) \cdot I,
    \end{equation}
    establishing $\min(\eta, \omega_k)$-strong convexity. 
\end{proof}

\begin{lemma}[Uniqueness of the Subproblem Solution] \label{lem: unique}
    Under Assumptions \ref{assump:system-properties}--\ref{assump constraint}, Problem \ref{prob: convex} has a unique minimizer.
\end{lemma}
 
\begin{proof}
    The constraints on Problem \ref{prob: convex} are: affine equality and inequality constraints (Equation \ref{eq: relaxation}) which are closed half-spaces and hyperplanes, the trust region is a closed infinity-norm constraint (Equation \ref{eq: trust region}), and the original convex constraints, $h_{\mathrm{cvx}}(z) \leq 0$, are convex by definition. Therefore, the feasible set is the intersection of convex sets which is a convex set. Consider the reference point in tandem with $\xi^i = D^i(\bar{z}^i)$ and $\zeta = [C(\bar{z})]_+$ is feasible (see the proof of Lemma \ref{lem: Delta L nonneg}). Therefore, the feasible set is not empty. Finally, applying Lemma \ref{lem: strongly convex}, $\Phi_k$ is strongly convex on $\mathcal{V}_k$, and a strongly convex function on a closed convex set attains a unique minimizer \cite{boyd2004convex}. This implies that Problem \ref{prob: convex} has a unique minimizer.
\end{proof}
 
\begin{lemma}[Equivalence to the Variational Nash Equilibrium] \label{lem: v-NE}
    Problem \ref{prob: convex} has a unique minimizer that is equivalent to the variational Nash equilibrium of the convexified game in which each agent $i$ minimizes its proximally-regularized augmented Lagrangian 
    \begin{equation}
        \tilde{L}_k^i(z^i, z^{-i}) := L_k^i(z^i, z^{-i}) + \tfrac{\eta}{2} \|z^i - \bar{z}_k^i\|^2
    \end{equation}
    while subject to the same linearized and convex constraints, with the shared-constraint multiplier common across agents.
\end{lemma}

\begin{proof}
    Consider the Lagrangian of Problem \ref{prob: convex} with the multipliers: $\nu^i_D$ for agent $i$'s dynamics-defect equality, $\nu_C$ for the single joint inequality $\tilde{C}(z) \leq \zeta$, $\nu_\zeta$ for $\zeta \geq 0$, $\nu^i_\tau$ for the epigraph, and $\nu^i_{\mathrm{TR}}$ for the trust region. The first order stationary conditions with respect to each block are then: 
    \begin{align}
        \partial_{z^i}: &\;\; \eta(z^i - \bar{z}^i) + \nu^i_\tau \nabla_{z^i} J^i + (A^i)^T \nu^i_D + (B^i)^T \nu_C + \nu^i_{\mathrm{TR}} + \partial h_{\mathrm{cvx}} = 0, \label{eq: stat z} \\
        \partial_{\xi^i}: &\;\; \lambda_k^i + \omega_k \xi^i - \nu^i_D = 0, \label{eq: stat xi} \\
        \partial_\zeta: &\;\; \mu_k + \omega_k \zeta - \nu_C - \nu_\zeta = 0, \label{eq: stat zeta} \\
        \partial_{\tau^i}: &\;\; 1 - \nu^i_\tau = 0, \label{eq: stat tau}
    \end{align}

    \noindent where $A^i = \nabla_{z^i} \tilde{D}^i$ and $B^i = \nabla_{z^i} \tilde{C}$ evaluated at $\bar{z}$. We can now substitute Equations \ref{eq: stat xi}--\ref{eq: stat zeta} into Equation \ref{eq: stat z} and notice that $\nu^i_\tau = 1$ yielding, for each agent $i$,
    \begin{equation} \label{eq: VI}
        \nabla_{z^i} J^i + (A^i)^T \!\big( \lambda_k^i + \omega_k \tilde{D}^i(z^i) \big) 
        + (B^i)^T \!\big( \mu_k + \omega_k [\tilde{C}(z)]_+ \big) 
        + \eta(z^i - \bar{z}^i) + \mathcal{N}_{\mathcal{K}^i}(z^i) \ni 0,
    \end{equation}
    where $\nu^i_{\mathrm{TR}} + \partial h_{\mathrm{cvx}} \in \mathcal{N}_{\mathcal{K}^i}$ is the normal cone to the feasible set of agent $i$. Therefore, Equation \ref{eq: VI} is the first-order condition $\nabla_{z^i} \tilde{L}_k^i(z^i, z^{-i}) \in -\mathcal{N}_{\mathcal{K}^i}(z^i)$ for each $i$, with a single common shared-constraint multiplier. It is important to note that the pair $(\nu_C, \nu_\zeta)$ is equivalent for each agent as enforced by how the Lagrangian is defined. Notice that there is one inequality $\tilde{C}(z) \leq \zeta$ and one nonnegativity constraint on $\zeta$. By definition, this common-multiplier property implies that the equilibrium determined is a variational Nash Equilibrium \cite{pavel2025operator}. From Lemma \ref{lem: unique}, the uniqueness of this equilibrium follows.    
\end{proof}

\subsection{Outer-Loop Convergence}

\begin{lemma}[Non-Negativity of the Predicted Cost Reduction] \label{lem: Delta L nonneg}
    Let $\Delta L_k := \sum_{i \in [N]} \Delta L_k^i$ denote the aggregate predicted cost reduction at an arbitrary iteration $k$. Then $\Delta L_k \geq 0$. Moreover, the following are equivalent:
    \begin{enumerate}
        \item $\Delta L_k = 0$;
        \item $\Delta L_k^i = 0$ for every $i \in [N]$;
        \item $\bar{z}_k$ is a stationary point of the penalized game, $Q_k$, with the current multipliers $(\lambda_k, \mu_k, \omega_k)$.
    \end{enumerate}
    $\max_i \Delta L_k^i \geq 0$, and $\max_i \Delta L_k^i = 0$ if and only if $\bar{z}_k$ is a stationary point of $Q_k$.
\end{lemma}
 
\begin{proof}
    \emph{Non-negativity of the sum.} Consider the candidate point 
    \begin{equation}
        v^\dagger := \big(\bar{z}_k,\, \{J^i(\bar{z}_k)\}_i,\, \{D^i(\bar{z}_k^i)\}_i,\, [C(\bar{z}_k)]_+\big).
    \end{equation}
    By direct substitution into the definitions of $\tilde{D}^i$ (Equation \ref{eq: D linear}) and $\tilde{C}$ (Equation \ref{eq: C linear}), the affine correction term vanishes at $z^i = \bar{z}_k^i$, resulting in $\tilde{D}^i(\bar{z}_k^i) = D^i(\bar{z}_k^i)$ and $\tilde{C}(\bar{z}_k) = C(\bar{z}_k) \leq [C(\bar{z}_k)]_+$, so $v^\dagger$ is a feasible point for Problem \ref{prob: convex}; the trust region is trivially satisfied. Evaluating $\Phi_k$ at $v^\dagger$ (the proximal term vanishes because $z = \bar{z}_k$):
    \begin{align}
        \Phi_k(v^\dagger) &= \sum_i \!\Big[ J^i(\bar{z}_k) + (\lambda_k^i)^T D^i(\bar{z}_k^i)  + \tfrac{\omega_k}{2}\|D^i(\bar{z}_k^i)\|^2 \Big] + (\mu_k)^T [C(\bar{z}_k)]_+ + \tfrac{\omega_k}{2}\|[C(\bar{z}_k)]_+\|^2 \notag \\
        &= \sum_i Q_k^i(\bar{z}_k). 
    \end{align}
    By the optimality of $v^*_k = (z^*_k, \tau^*_k, \xi^*_k, \zeta^*_k)$, $\Phi_k(v^*_k) \leq \Phi_k(v^\dagger)$. Unrolling $\Phi_k(v^*_k) = \sum_i L_k^i(z^*_k, \xi_k^{i,*}, \zeta^*_k)$ results in 
    \begin{equation}
        \sum_i L_k^i(v^*_k) \leq \sum_i Q_k^i(\bar{z}_k),
    \end{equation}
    i.e., $\Delta L_k = \sum_i \Delta L_k^i \geq 0$.
    
    \emph{Equivalence of (1), (2), (3).}  
    \textbf{(1 $\Rightarrow$ 3):} By strong convexity (Lemma \ref{lem: strongly convex}), the $v^*_k$ is the unique minimizer, so $\Phi_k(v^*_k) = \Phi_k(v^\dagger)$ implies $v^*_k = v^\dagger$. Moreover, $z^*_k = \bar{z}_k$, $\xi_k^{i,*} = D^i(\bar{z}_k^i)$, and $\zeta^*_k = [C(\bar{z}_k)]_+$. The aforementioned equivalences can be substituted into the KKT conditions of Problem \ref{prob: convex} (Equations \ref{eq: stat z}--\ref{eq: stat tau}) with $\nu^i_D = \lambda_k^i + \omega_k D^i(\bar{z}_k^i)$ and $\nu_C = \mu_k + \omega_k [C(\bar{z}_k)]_+$ yields, for each $i \in [N]$,
    \begin{equation} \label{eq: per agent stationarity}
        \nabla_{z^i} Q_k^i(\bar{z}_k) + \mathcal{N}_{\mathcal{K}^i}(\bar{z}_k^i) \ni 0,
    \end{equation}
    which, for each agent $i$, is the equivalent to the first-order necessary condition of $\min_{z^i \in \mathcal{K}^i} Q_k^i(z^i, \bar{z}_k^{-i})$ by \cite[Thm.~12.1]{nocedal2006numerical} applied to the individual problem for agent $i$. LICQ for each $\mathcal{K}^i$ is inherited from \cite[Asn.~1]{oguri2023successive}. When $\bar{z}_k$ is a regular point of the active-constraint set the penalty function, $P_k$, is $\mathcal{C}^1$ in $z^i$ and the gradient $\nabla_{z^i} Q_k^i$ in Equation \ref{eq: per agent stationarity} is the ordinary Fr\'echet gradient. The Clarke generalized-gradient form of the Lagrange multiplier rule \cite[Thm.~6.1.1 \& Rem.~6.1.2 (iv)]{clarke1990optimization} handles this case where $[\,\cdot\,]_+$ is active. The necessary condition becomes $0 \in \partial_C Q_k^i(\bar{z}_k) + \mathcal{N}_{\mathcal{K}^i}(\bar{z}_k^i)$ via the Clarke subdifferential $\partial_C$. The collection of Equation \ref{eq: per agent stationarity} over $i \in [N]$ is, by definition, a stationary point of the penalized game $Q_k$ at $\bar{z}_k$.  
 
    \textbf{(3 $\Rightarrow$ 2):} If $\bar{z}_k$ is a stationary point of $Q_k$, then Equation \ref{eq: per agent stationarity} holds for every $i$. Following the same KKT-to-stationarity equivalence of (1 $\Rightarrow$ 3) in the opposite direction demonstrates that $v^\dagger$ satisfies the KKT conditions of Problem \ref{prob: convex} with equivalent dual choices. Therefore, by strong convexity and uniqueness of the minimizer as shown in Lemmas \ref{lem: strongly convex} and \ref{lem: unique}, $v^*_k = v^\dagger$. Moreover, for every $i$,
    \begin{equation}
        L_k^i(v^*_k) = L_k^i(v^\dagger) = Q_k^i(\bar{z}_k),
    \end{equation}
    so $\Delta L_k^i = 0$ for every $i$.  
 
    \textbf{(2 $\Rightarrow$ 1):} Immediate by summation.
    
    \emph{The max-statement.} Since $\sum_i \Delta L_k^i \geq 0$, $\max_i \Delta L_k^i \geq 0$. For the equality: $\max_i \Delta L_k^i = 0$ combined with $\sum_i \Delta L_k^i \geq 0$ forces $\Delta L_k^i \leq 0$ for all $i$ and $\sum_i \Delta L_k^i \leq 0$; combined with the lower bound, $\sum_i \Delta L_k^i = 0$, which is case (1) above, so $\Delta L_k^i = 0$ for all $i$ by the equivalence, and stationarity of $\bar{z}_k$ follows.
\end{proof}

\begin{lemma}[Finite Acceptance] \label{lem: finite acceptance}
    Under Assumptions \ref{assump:system-properties}--\ref{assump constraint}, any rejected iteration is followed by an accepted iteration within a finite number of trust-region contractions.
\end{lemma}
 
\begin{proof}
    Along a sequence of repeated rejections, $r_k$ contracts geometrically by $\alpha_1^{-1} < 1$ until bounded below by $r_{\min}$. By $\mathcal{C}^2$ smoothness of $D^i$ and $C$ via Assumptions \ref{assump: cost}--\ref{assump constraint}, the linearization errors $|Q_k^i(z^*_k) - L_k^i(z^*_k, \xi^{i,*}_k, \zeta^*_k)|$ are $O(r_k^2)$. Additionally, the predicted reduction $\Delta L_k^i$ is $O(r_k)$ if $\bar{z}_k$ is non-stationary. Hence,
    \begin{equation}
        \rho_k = 1 - O(r_k) \to 1 \text{ as } r_k \to 0,
    \end{equation}
    so the acceptance threshold $\rho_0 < 1$ is eventually crossed.
\end{proof}

\begin{remark}
    The aforementioned Lemma and proof mirrors Lemma 3.11 of \cite{mao2016successive} and Lemma 4 of \cite{oguri2023successive}. 
\end{remark}

\begin{lemma}[Limit Points under Frozen Multipliers] \label{lem: limit points}
    Suppose the multipliers $(\lambda_k, \mu_k, \omega_k)$ are held fixed at $(\lambda_{k_0}, \mu_{k_0}, \omega_{k_0})$ for all $k \geq k_0$. Let $\{\bar{z}_k\}_{k \geq k_0}$ denote the sequence of accepted iterates. Then:
    \begin{enumerate}
        \item $\{\bar{z}_k\}$ admits at least one limit point in the compact set $\mathcal{K} = \prod_i \mathcal{K}^i$;
        \item $\max_i |\Delta Q_k^i| \to 0$ as $k \to \infty$;
        \item every limit point is a stationary point of the penalized game $Q_{k_0}$.
    \end{enumerate}
\end{lemma}
 
\begin{proof}
    Step acceptance gives $\min_i \Delta Q_k^i \geq \rho_0 \max_i \Delta L_k^i \geq 0$, where the last inequality is from Lemma \ref{lem: Delta L nonneg}. Hence $\Delta Q_k^i \geq 0$ for every $i$ and every $k \geq k_0$. Since the multipliers are frozen, $Q_k^i = Q_{k_0}^i$ as functions, and $\sum_i Q_{k_0}^i(\bar{z}_k)$ is monotonically non-increasing in $k$. It is bounded below by zero (the penalty terms are non-negative and $\tau^i$ is bounded below on the compact feasible set by Assumption \ref{assump: cost}). Hence the sequence converges, and telescoping gives $\sum_{k \geq k_0} \Delta Q_k^i < \infty$, which implies $\max_i |\Delta Q_k^i| \to 0$, establishing (2). Boundedness of the iterates in the compact $\mathcal{K}$ together with Bolzano--Weierstrass gives (1) \cite{bartle2000introduction}. For (3), suppose a limit point $\hat{z}$ is not stationary: then some agent $i_0$ has a strictly profitable feasible deviation $v^{i_0}$ with $Q_{k_0}^{i_0}(v^{i_0}, \hat{z}^{-i_0}) < Q_{k_0}^{i_0}(\hat{z}) - \theta_0$ for some $\theta_0 > 0$. By $\mathcal{C}^2$ continuity this persists in a neighborhood of $\hat{z}$, so for all sufficiently large $k$ along a subsequence to $\hat{z}$, $\Delta Q_k^{i_0} \geq \theta_0 / 2$, contradicting (2).
\end{proof}

\begin{lemma}[Infinite Multiplier Updates] \label{lem: infinite updates}
    The multiplier update condition Equation \ref{eq: dual update condition} is satisfied at an infinite subsequence of iterations $\{k_j\}_{j \geq 1}$, and $\delta_{k_j} \to 0$.
\end{lemma}
 
\begin{proof}
    Suppose, for contradiction, that after some iteration $k_0$ the condition is never met. Then multipliers are frozen for all $k \geq k_0$, and Lemma \ref{lem: limit points} (2) gives $\max_i |\Delta Q_k^i| \to 0$. But $\delta_k > 0$ is fixed after $k_0$ (decreased only on updates), so the condition is eventually satisfied, a contradiction. By Equation \ref{eq: delta update}, $\delta_{k_j} = \gamma_\delta^{\,j} \delta_{k_0} \to 0$.
\end{proof}

\begin{theorem}[Global Strong Convergence with Feasibility] \label{thm: main convergence}
    Under Assumptions \ref{assump:system-properties}--\ref{assump constraint}, and provided the second-order sufficient conditions of \cite[Assumption 1]{oguri2023successive} hold at the limit, Algorithm \ref{alg: FALCON} converges globally to a point $z^\star$ that is:
    \begin{enumerate}
        \item a variational Nash equilibrium of the original non-convex PDGNEP of Problem \ref{prob: general formulation};
        \item dynamically feasible, $D^i(z^{i,\star}) = 0$ for all $i$;
        \item approximately coupling-feasible in the continuous-time sense,
        \begin{equation}
            C_\rho(x^\star(t), u^{i,\star}(t)) \leq \delta_\rho(\varepsilon) \to 0 \text{ as } \varepsilon \to 0,
        \end{equation}
        for every $t \in [t_k, t_{k+1}]$ and every coupling constraint component $\rho$.
    \end{enumerate}
\end{theorem}
 
\begin{proof}
    By Lemma \ref{lem: infinite updates}, there is an infinite subsequence $\{k_j\}$ of multiplier-update iterations with $\delta_{k_j} \to 0$, so $\max_i |\Delta Q^{i}_{k_j}| \to 0$ along this subsequence. By Lemma \ref{lem: limit points} (3), every limit point of the accepted iterates under frozen multipliers is a stationary point of the corresponding penalized game, which by Lemma \ref{lem: v-NE} is the v-NE of the penalized game with multipliers $(\lambda_{k_j}, \mu_{k_j}, \omega_{k_j})$. First-order optimality \cite[Thm. 12.1]{nocedal2006numerical} applied to this v-NE yields $\|\nabla_z L_{\omega_{k_j}}(z_{k_j}^*, \lambda_{k_j}, \mu_{k_j})\|_2 \to 0$, where the affine and convex constraints (the trust region, $\zeta \geq 0$, and $h_{\mathrm{cvx}}(z) \leq 0$) are satisfied exactly by construction of Problem \ref{prob: convex}. 
 
    The sequence $\{(z_{k_j}^*, \lambda_{k_j}, \mu_{k_j})\}$ therefore satisfies the hypotheses of Lemma \ref{lem: AL convergence}: asymptotic exactness holds, the multipliers are bounded under the stated second-order sufficient condition, and $\omega_k$ increases geometrically. Applying Lemma \ref{lem: AL convergence}, the sequence globally converges to $(z^\star, \lambda^\star, \mu^\star)$ satisfying the KKT conditions of Problem \ref{prob: general formulation}.
 
    The KKT conditions of Problem \ref{prob: general formulation} at the limit enforce, per agent, stationarity $\nabla_{z^i}[J^i + (\lambda^{i,\star})^T D^i + (\mu^\star)^T C] = 0$, primal feasibility $D^i(z^{i,\star}) = 0$ and $[C(z^\star)]_+ = 0$, and complementary slackness. Stationarity is the condition that $z^{i,\star}$ is a local minimizer of agent $i$'s problem given $z^{-i,\star}$; combined with the common-$\mu^\star$ structure that arises from $\mu_0^i = 0$ for all $i$, this is exactly the definition of a v-NE of Problem \ref{prob: general formulation}, giving (1). Primal feasibility $D^i(z^{i,\star}) = 0$ gives (2). Finally, the condition $[C(z^\star)]_+ = 0$ applied to the augmented coupling constraint (which includes the continuous-time boundary constraint Equation \ref{eq: CT bc}) together with Theorem \ref{thm: elango 14} gives the pointwise bound in (3).
\end{proof}

\begin{remark}
    FALCON is globally convergent in the sense that, from any initialization satisfying the algorithm assumptions, the generated sequence converges to a feasible open-loop variational Nash equilibrium of the original PDGNEP.
\end{remark}

\subsection{Convergence Rate} \label{subsec: rate}
 
Theorem \ref{thm: main convergence} establishes that Algorithm \ref{alg: FALCON} converges globally, but the convergence rate is not yet considered. Since FALCON inherits its outer-loop structure from the augmented Lagrangian method, classical rate results for augmented Lagrangian iterations \cite[Prop. 2.4]{bertsekas2014constrained} carry over to the multiplier-update subsequence with a minor strengthening of the update criterion. This subsection makes that argument precise. Throughout, the rate is characterized \emph{on the subsequence $\{k_j\}$ at which multipliers are updated}, not on the outer iteration index $k$ itself.
 
\subsubsection{A Stricter Multiplier Update Criterion}
 
The standard multiplier update condition Equation \ref{eq: dual update condition} forces $\max_i |\Delta Q_k^i| < \delta_k$ with $\delta_k \to 0$, which is enough for asymptotic convergence but does not tie the inexactness tolerance to the primal feasibility residual. A stricter criterion that does so is
\begin{equation} \label{eq: dual update condition strict}
    \max_i |\Delta Q_k^i| \leq \min\!\left\{ \delta_k,\, \eta_{\rm rate}\, \chi_k \right\},
\end{equation}
where $\eta_{\rm rate} > 0$ is a user-chosen constant and $\chi_k$ is the primal feasibility residual defined in Equation \ref{eq: cost updates}. The criterion Equation \ref{eq: dual update condition strict} is the direct multi-agent analog of Oguri's Eq. (20) \cite[\S IV-B]{oguri2023successive}; we use the subscript ``rate'' on $\eta_{\rm rate}$ to disambiguate it from the proximal parameter $\eta$ of Problem \ref{prob: convex}.
 
\begin{lemma}[Finite Multiplier Updates under the Stricter Criterion] \label{lem: finite updates strict}
    Suppose the multiplier update criterion Equation \ref{eq: dual update condition strict} is used in place of Equation \ref{eq: dual update condition}. Until convergence of $(\lambda_k, \mu_k)$ is achieved, the multipliers are updated within a finite number of iterations after each previous update.
\end{lemma}
 
\begin{proof}
    By Lemma \ref{lem: infinite updates} (applied to the weaker criterion Equation \ref{eq: dual update condition}), multiplier updates occur infinitely often when the condition $\max_i |\Delta Q_k^i| < \delta_k$ alone is used. It suffices, therefore, to show that the additional conjunct $\max_i |\Delta Q_k^i| \leq \eta_{\rm rate}\, \chi_k$ of Equation \ref{eq: dual update condition strict} does not prevent updates, i.e., $\chi_k > 0$ whenever $(\lambda_k, \mu_k) \neq (\lambda^\star, \mu^\star)$.
    
    Suppose, for contradiction, that $(\lambda_k, \mu_k) \neq (\lambda^\star, \mu^\star)$ and $\chi_k = 0$. From Equation \ref{eq: cost updates}, $\chi_k = 0$ implies $D^i(z_k^*) = 0$ for every $i$ and $[C(z_k^*)]_+ = 0$; that is, $z_k^*$ is primal-feasible for Problem \ref{prob: general formulation}. Combined with $\max_i |\Delta Q_k^i| \to 0$ (Lemma \ref{lem: limit points} (2) under frozen multipliers), primal feasibility and asymptotic stationarity together form the KKT conditions of Problem \ref{prob: general formulation} at $(z_k^*, \lambda_k, \mu_k)$. By uniqueness of multipliers under the second-order sufficient conditions \cite[Asn.~1]{oguri2023successive}, this forces $(\lambda_k, \mu_k) = (\lambda^\star, \mu^\star)$, contradicting the hypothesis. Hence $\chi_k > 0$, and Equation \ref{eq: dual update condition strict} is eventually satisfied, triggering a multiplier update.
\end{proof}

\subsubsection{Rate Theorem}
 
\begin{theorem}[Convergence Rate of FALCON] \label{thm: rate}
    Suppose Algorithm \ref{alg: FALCON} is evaluated with the multiplier update criterion in Equation \ref{eq: dual update condition strict} instead of Equation \ref{eq: dual update condition}, and let $\{k_j\}_{j \geq 1}$ denote the subsequence of iterations at which the multipliers are updated. Let $(z^\star, \lambda^\star, \mu^\star)$ be the limit guaranteed by Theorem \ref{thm: main convergence}. Then:
    \begin{enumerate}
        \item If the penalty weight is capped at $\omega_k \leq \omega_{\max} < \infty$ for all $k$, the multiplier sequence converges \emph{Q-linearly}:
        \begin{equation}
            \| (\lambda_{k_{j+1}}, \mu_{k_{j+1}}) - (\lambda^\star, \mu^\star) \|_2 \leq \kappa_{\rm rate}\, \| (\lambda_{k_j}, \mu_{k_j}) - (\lambda^\star, \mu^\star) \|_2
        \end{equation}
        for some $\kappa_{\rm rate} \in (0, 1)$ and all sufficiently large $j$.
        \item If $\omega_k \to \infty$ (as under the update rule Equation \ref{eq: dual updates}c with $\beta > 1$), the convergence is \emph{Q-superlinear}:
        \begin{equation}
            \lim_{j \to \infty} \frac{\| (\lambda_{k_{j+1}}, \mu_{k_{j+1}}) - (\lambda^\star, \mu^\star) \|_2}{\| (\lambda_{k_j}, \mu_{k_j}) - (\lambda^\star, \mu^\star) \|_2} = 0.
        \end{equation}
    \end{enumerate}
    In both cases, $z_{k_j}^* \to z^\star$ at a rate no slower than that of the multiplier iteration, and exactly matching it in the first-order sensitivity sense \cite[Prop. 2.7]{bertsekas2014constrained}.
\end{theorem}
 
\begin{proof}
    We apply the augmented Lagrangian convergence rate theory of \cite[\S 2.4]{bertsekas2014constrained} to the multiplier-update subsequence $\{k_j\}$. The second-order sufficient conditions of \cite[Asn. 1]{oguri2023successive} hold at the limit (inherited from Theorem \ref{thm: main convergence}).
    
    \emph{Implicit-function regularity}
    
    Consider the inner problem at the $j$-th multiplier update:
    \begin{equation}
        \min_{z} L_{\omega_{k_j}}(z, \lambda_{k_j}, \mu_{k_j}) \quad \text{s.t.} \quad z \in \mathcal{K},
    \end{equation}
    where $\mathcal{K} = \prod_i \mathcal{K}^i$ is the feasible set. Under the second-order sufficient conditions, the Hessian $\nabla_z^2 L_{\omega_{k_j}}(z^\star, \lambda^\star, \mu^\star)$ is positive definite on the subspace of feasible directions. 
    
    By \cite[Prop. 2.4]{bertsekas2014constrained}, for all $((\lambda, \mu), \omega)$ in the neighborhood of $((\lambda^\star, \mu^\star), \omega_{\max})$, there exists a unique solution $z(\lambda, \mu, \omega)$ to the inner problem such that
    \begin{equation} \label{eq: implicit bound}
        \| z(\lambda, \mu, \omega) - z^\star \| \leq M \frac{\| (\lambda, \mu) - (\lambda^\star, \mu^\star) \|}{\omega}
    \end{equation}
    where $M > 0$. The minimizer $z_{k_j}^* = z(\lambda_{k_j}, \mu_{k_j}, \omega_{k_j})$ is well-defined for all sufficiently large $j$.
    
    \emph{Multiplier update as a fixed-point iteration}
    
    We define the multiplier error, $e_j := (\lambda_{k_j}, \mu_{k_j}) - (\lambda^\star, \mu^\star)$, and let $h := (D^i, [C]_+)$. The dual update in Equation \ref{eq: dual updates} can be written
    \begin{equation}
        e_{j+1} = e_j + \omega_{k_j} [h(z_{k_j}^*) - h(z^\star)],
    \end{equation}
    where we used $h(z^\star) = 0$ (primal feasibility from Theorem \ref{thm: main convergence}). By the implicit bound in Equation \ref{eq: implicit bound} and the Lipschitz continuity of $h$ on the compact feasible set with constant $L > 0$,
    \begin{equation}
        \| h(z_{k_j}^*) - h(z^\star) \| \leq L \| z_{k_j}^* - z^\star \| \leq \frac{LM}{\omega_{k_j}} \| e_j \|.
    \end{equation}
    The criterion Equation \ref{eq: dual update condition strict} ensures $\|\nabla_z L_{\omega_{k_j}}\|_2 = O(\chi_{k_j})$, the inexact-minimization condition required by the rate theorem.
    
    \emph{Q-linear rate (capped penalty)}
    
    If $\omega_{k_j} \leq \omega_{\max} < \infty$, by \cite[Prop. 2.7]{bertsekas2014constrained} the contraction factor takes the form
    \begin{equation}
        \kappa_{\rm rate} = \max_{i=1,\ldots,m} \frac{|e_i|}{e_i + \omega_{\max}},
    \end{equation}
    where $e_1, \ldots, e_m$ are the eigenvalues of $\nabla^2 p(0)$ at the limit (in the notation of \cite{bertsekas2014constrained}). For sufficiently large $\omega_{\max}$, this ratio is strictly less than 1, giving Q-linear convergence.
    
    \emph{Q-superlinear rate (unbounded penalty)}
    
    If $\omega_{k_j} \to \infty$, then by \cite[Prop. 2.7, Eq. (44)]{bertsekas2014constrained},
    \begin{equation}
        \lim_{j \to \infty} \frac{\| e_{j+1} \|}{\| e_j \|} = 0,
    \end{equation}
    which is Q-superlinear convergence.
    
    \emph{Primal convergence rate}
    
    By Equation \ref{eq: implicit bound}, $\| z_{k_j}^* - z^\star \| \leq M \| e_j \| / \omega_{k_j}$, so the primal error decays at the same rate as the multiplier error (capped $\omega$) or faster (unbounded $\omega$). This is the first-order sensitivity statement of \cite[Prop. 2.7]{bertsekas2014constrained}.
\end{proof}

\section{Numerical Results} \label{sec: Numerical Results}

We consider three application problems as demonstrators for the FALCON algorithm with each problem demonstrating a form of increasing difficulty. First, a racing game on a sharp corner of the F1 Austin track that is used to benchmark against other methods. Second, a narrowing hallways problem that demonstrates a difficult, collaborative multi-agent path planning problem wherein the players must maintain safety constraints while remaining within a communication range planned in SE(2). Finally, a spacecraft lady-bandit-guard game that simulates attacking and defending a high value space asset planned in a nonholonomic $\mathbb{R}^6$ space. All of these examples were created using the DifferentialGames.jl ecosystem \cite{outland_2026_20432309}. 

\subsection{Austin F1 Racing}

\begin{figure}
    \centering
    \includegraphics[width=\linewidth]{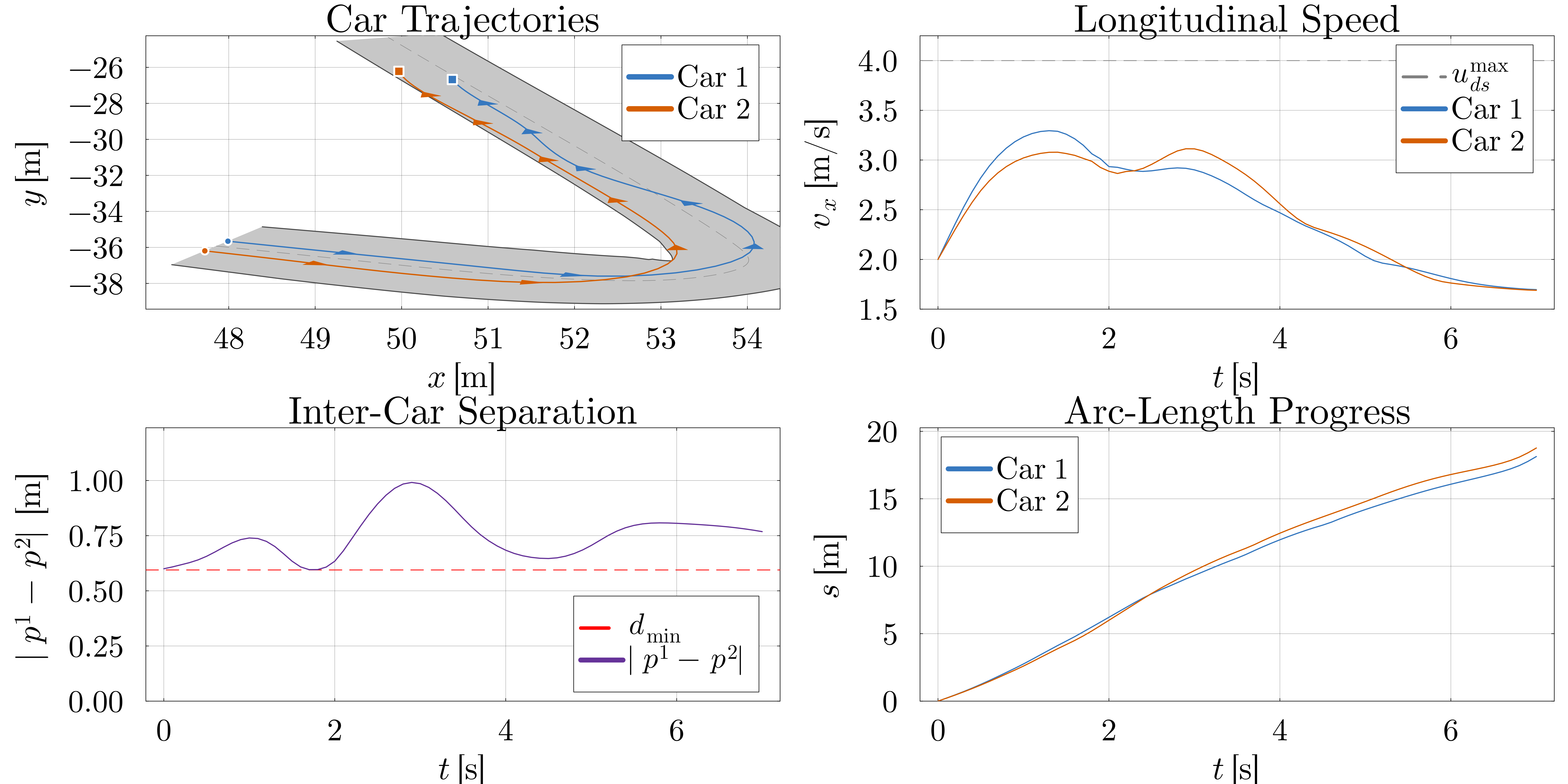}
    \caption{A two-player game planned in $\text{SE}(2)$ in which the players race through a sharp corner on the F1 Austin track. The circles denote the initial conditions, the square denotes the terminal position at the end time horizon.}
    \label{fig:f1}
\end{figure}

\begin{figure}
    \centering
    \includegraphics[width=\linewidth]{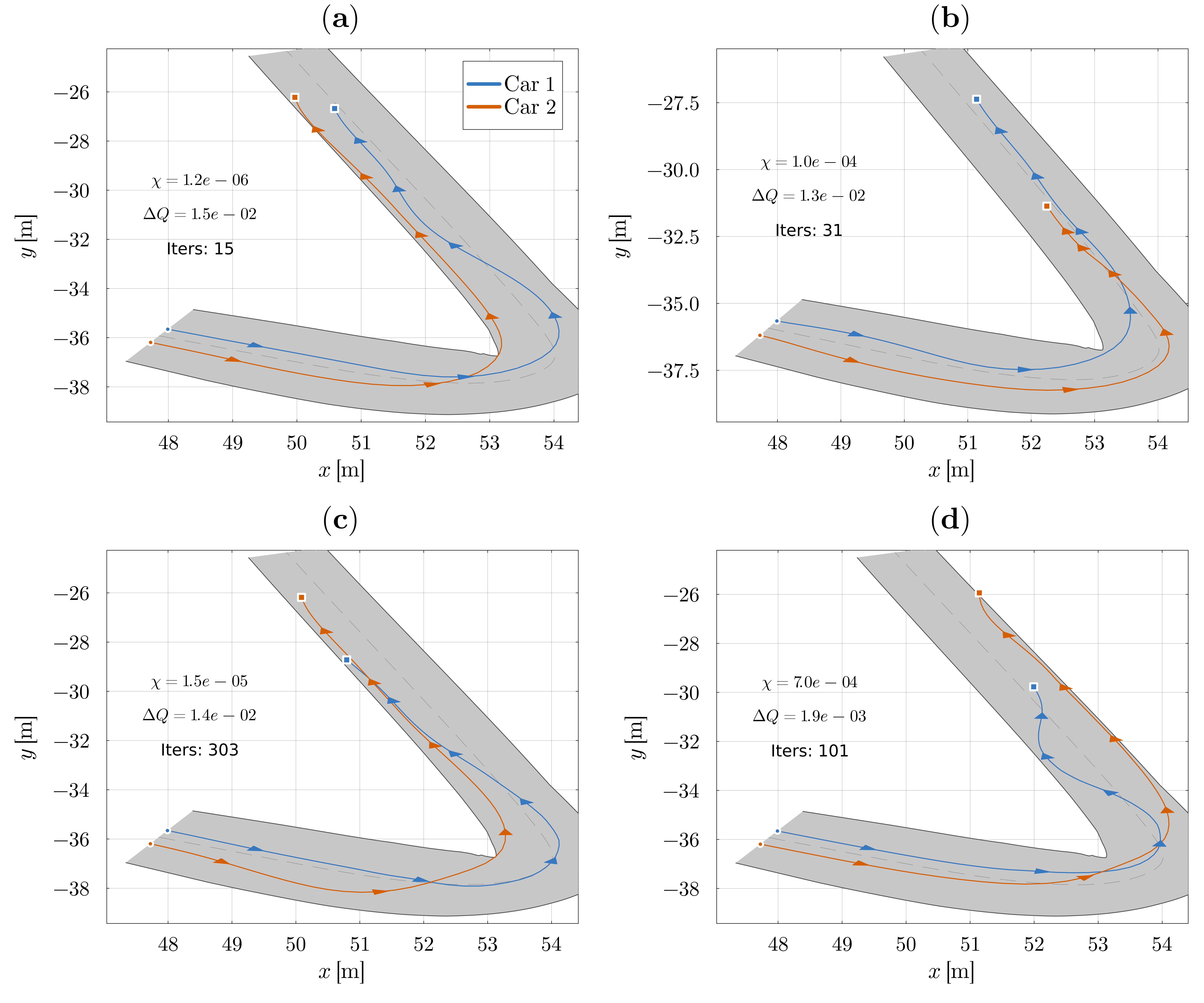}
    \caption{Permutations of the initial velocities for the F1 Austin problem. Sub-figure a corresponds to the simulation in Figure \ref{fig:f1}. The convergence certificate values are given o n the left of each sub-figure. The circles denote the initial conditions and the square denotes the position at the end of the time horizon.}
    \label{fig:f1_perm}
\end{figure}

For comparison against ALGAMES \cite{le2022algames} and DG-SQP \cite{zhu2023sequential} we consider an F1 racing problem between two players on a sharp corner of the Austin F1 track where the cars are modeled as dynamic bicycles \cite{zhu2024sequential}. Examples of what this corner looks like can be seen in Figures \ref{fig:f1} and \ref{fig:f1_perm}. This game is solved in the Frenet frame \cite{zhu2024sequential}. Consider a parametric path, $\mathcal{T}:[0, L] \rightarrow \mathbb{R}^2$, of length $L$. There is an arclength, $s \in [0, L]$, with a tangent angle of $\phi(s) = \arctan(\mathcal{T}_y'(s) / \mathcal{T}_x'(s))$. Given a position, $p=(x, y)$, and yaw angle, $\psi$, the position in the Frenet frame is 

\begin{subequations}
    \begin{equation}
        s(p) = \argmin_s \| \mathcal{T}(s) - p \|
    \end{equation}
    \begin{equation}
        e_y(p) = [\sin \phi(s(p)), \cos \phi(s(p))](p - \mathcal{T}(s(p))
    \end{equation}
    \begin{equation}
        e_\psi(p, \psi) = \psi - \phi(s(p)), 
    \end{equation}
\end{subequations}

\noindent which are the progress along the path, displacement laterally from the path, and the heading deviation. The path progress is replaced by approximated version, $\bar{s}$ introduced in previous work \cite{zhu2024sequential}. The quality of this approximation is then evaluated using a lag error function, 

\begin{equation}
    e_l(p^i, \bar{s}^i) = \left[ - \cos \phi(\bar{s}^i), - \sin \phi(\bar{s}^i)  \right] (p^i - \mathcal{T}(\bar{s}^i)).
\end{equation}

\noindent This lag error is used as part of the cost function,

\begin{equation}
    J^i = V^i + \sum_{k=1}^T q_l e_l(p^i, \bar{s}^i)^2,
\end{equation}

\noindent where $V^i$ is a set of many different quadratic objectives for behavioral tuning and $q_l \in \mathbb{R}_+$ is a tuning parameter for penalizing the lag error.

In addition to control constraints and arc speed bounds (corners can only be taken so fast), the track boundary conditions are defined as 

\begin{equation}
    -w^-(\bar{s}^i_k) \leq e_c(p_k^i, \bar{s}^i_k) \leq -w^+(\bar{s}^i_k),
\end{equation}
where $w^-, w^+$ are distance metrics to each side of the track.

The numerical results can be seen in Figure \ref{fig:f1}. These results correspond with the types of trajectories demonstrated in the original paper \cite{zhu2024sequential}. However, while the cars still remain head-to-head, we change the initial velocities of the cars and display some of the permutations in Figure \ref{fig:f1_perm}. As expected, across all figures, we see both cars fighting to be interior to the turn in order to cut off the other player. In subfigure d of Figure \ref{fig:f1_perm}, we see a case where car 1 attempted to cut off car 2, but failed to do so.

We benchmarked the FALCON algorithm against the ALGAMES \cite{le2022algames} and DG-SQP \cite{zhu2023sequential} algorithms on this problem and the results are listed in Table \ref{tab:dgsqp_falcon}. Note, ALGAMES failed to converge for this difficult case when using the implementation from \cite{zhu2023sequential} when held to the same convergence standards as FALCON and DG-SQP. As observed by \cite{zhu2023sequential} and \cite{zhu2024sequential}, ALGAMES struggles with tight corners. For the timing and iteration data presented for DG-SQP and FALCON, both methods were warm started by optimal control strategies on a single agent basis. For example, following the code from the DG-SQP paper, DG-SQP was warm-started using mode predictive control. Likewise, FALCON is applied to a single agent (since it reduces to Scvx*) for the warm-start. The FALCON algorithm demonstrates very strong performance in comparison to DGSQP. While making variations to the initial conditions, DGSQP only converged 72\% of the time, while FALCON converged 100\% of the time. Furthermore, we can see that FALCON is significantly faster than DG-SQP. While some of this is due to FALCON and DG-SQP being programmed in different languages, we find that FALCON needs fewer iterations as well. We include both the computational time for the cases where the solver converged and all trials as performed on a AMD Ryzen 7 7445HS. Since FALCON converges for all of the trials, these values are the same. Ultimately, FALCON has significantly better performance in both convergence and computational speed.

\stepcounter{table}
\begin{table}
    \caption{Performance comparison of DG-SQP and FALCON for a 2-player F1 Austin racing scenario ($N=25$, $\Delta t = 0.1\,\text{s}$, $M=10$, $\chi = 1$e-3, $\Delta Q = 2$e-2).}
    \label{tab:dgsqp_falcon}
    \centering
    \renewcommand{\arraystretch}{1.2}
    \setlength{\tabcolsep}{4pt}
    \begin{tabularx}{\linewidth}{@{}l *{4}{>{\centering\arraybackslash}X}@{}}
        \hline
        Algorithm & Convergence Rate (\%) & Time (s) & Iterations & Time All (s) \\
        \hline
        DG-SQP \cite{zhu2023sequential} & 72.0 & 50.021 $\pm$ 19.269 (31.696--124.621) & 422.8 $\pm$ 138.0 (276--947) & 81.471 $\pm$ 131.536 \\
        FALCON (ours) & \textbf{100.0} & \textbf{0.873 $\pm$ 0.403 (0.200--2.318)} & \textbf{85.9 $\pm$ 43.8 (17--252)} & \textbf{0.873 $\pm$ 0.403} \\
        \hline
    \end{tabularx}
\end{table}

\subsection{Narrowing Hallways Problem}

\begin{figure}
    \centering
    \includegraphics[width=0.85\linewidth]{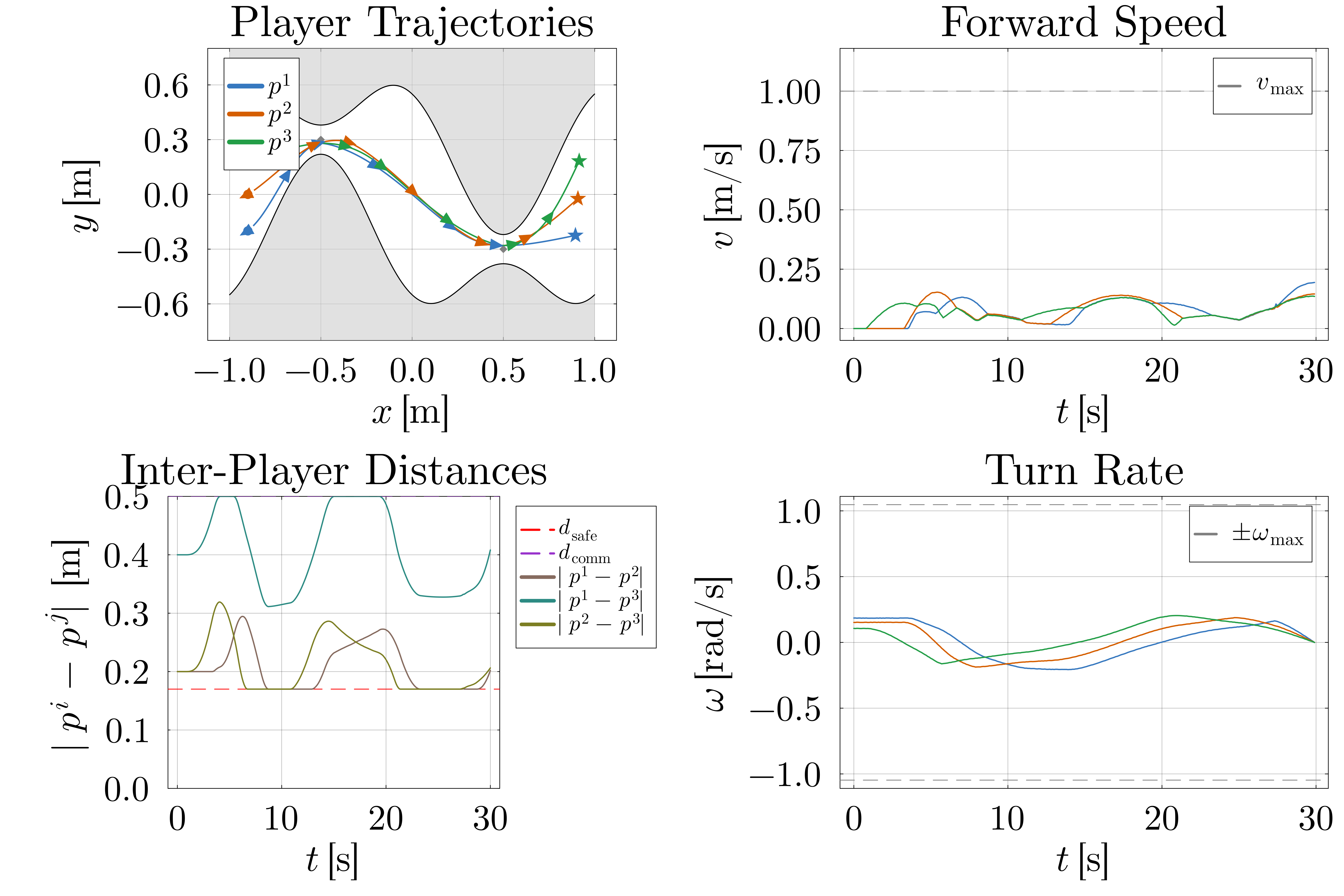}
    \caption{A three-player game planned in $\text{SE}(2)$ in which the players must navigate through a hallway with two choke-points. The circles denote the initial conditions, the star  denotes the goal position, and the diamond indicates the center of the choke-points.}
    \label{fig:hallway}
\end{figure}

\begin{figure}
    \centering
    \includegraphics[width=\linewidth]{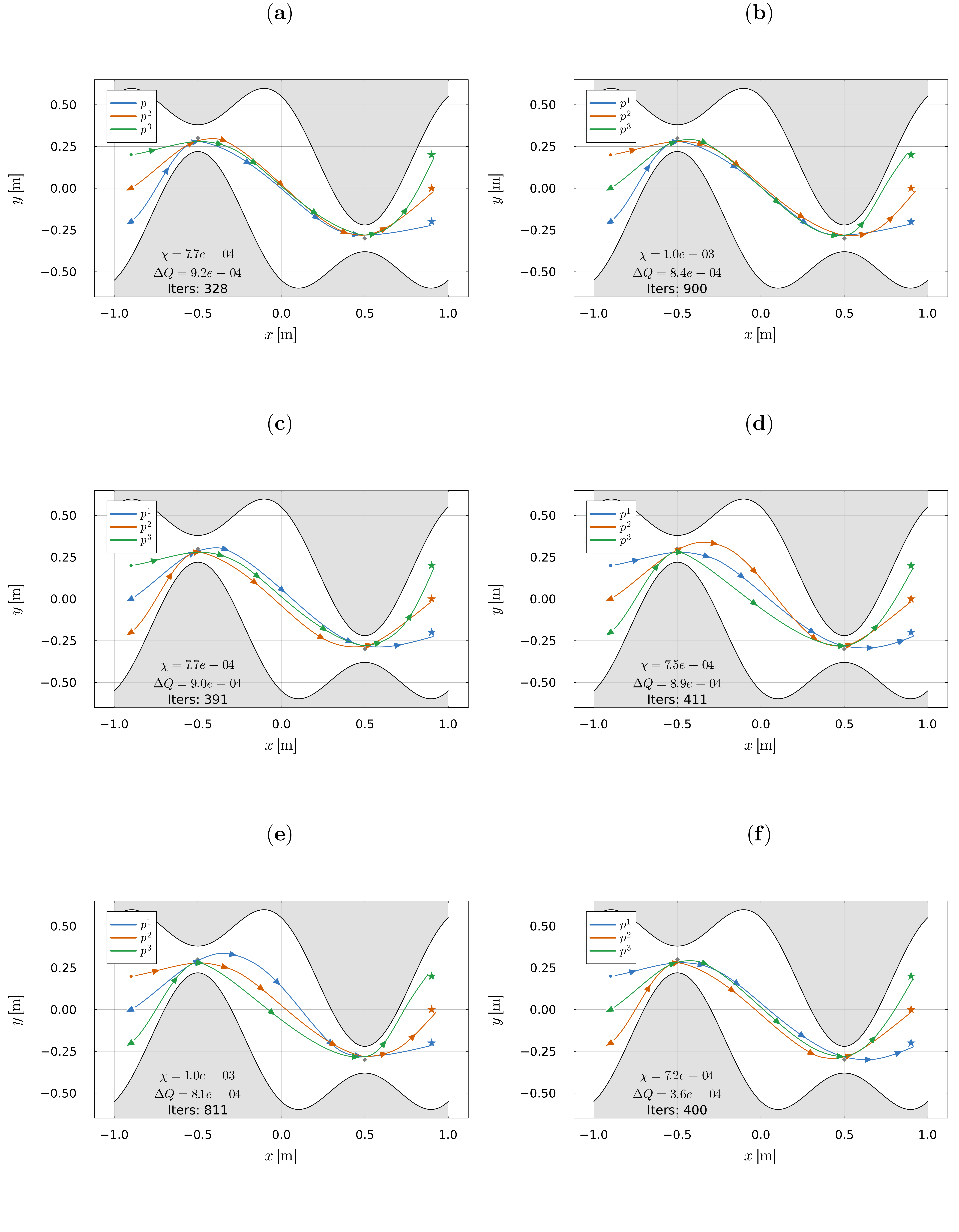}
    \caption{Permutations of the initial conditions and goal states for the narrowing hallways problem. Sub-figure a corresponds to the simulation in Figure \ref{fig:hallway}. The convergence certificate values are given in the bottom left corner of each sub-figure. The circles denote the initial conditions, the star  denotes the goal position, and the diamond indicates the center of the choke-points.}
    \label{fig:hallway_perm}
\end{figure}

We consider a scenario wherein three players controlling unicycles need to reach a goal state while navigating a set of hallways with two choke-points. An example configuration can be seen in Figure \ref{fig:hallway} and initial condition and goal state permutations can be seen in Figure \ref{fig:hallway_perm}. Each player $i$ minimizes
\begin{equation}
    J^i = 
        \sum_{k=0}^{T} \Bigl[
           (v^i_k)^2 + (\omega^i_k)^2
          \Bigr],
\end{equation}
where $g^i \in \mathbb{R}^3$ is the goal state of player $i$, $\rho_{\mathrm{goal}}, \rho_c, \rho_p,$ and $\rho_{\mathrm{comm}}$ are weights, $d_p$ is the proximity repulsion activation radius, and $d_{\mathrm{comm}}$ is the maximum communication range. At all times, the players must uphold four constraints. 

First, a wall constraint:
\begin{equation}
    y_c(p_{k,x}^{i}) - \tilde{w}(p_{k,x}^{i})
    \;\leq\; p_{k,y}^{i}
    \;\leq\; y_c(p_{k,x}^{i}) + \tilde{w}(p_{k,x}^{i}),
\end{equation}
where $\tilde{w}(p_x) = w(p_x) - d_c$ is the effective half-width after a safety inset $d_c = d_{\mathrm{wall}} + d_{\mathrm{buffer}}$. 

Second, an inter-player collision constraint modeled by radius $d_s$,
\begin{equation}
    \lVert p^i_k - p^j_k \rVert \geq d_s,
    \qquad \forall\,i < j,\; \forall\,k.
\end{equation}

Third, a communication-range constraint,
\begin{equation}
    \lVert p^i_k - p^j_k \rVert \leq d_{\mathrm{comm}},
    \qquad \forall\,i < j,\; \forall\,k.
\end{equation}

Finally, a terminal constraint,
\begin{equation}
    \lVert p^i_T - g_i \rVert \leq \epsilon_{goal}.
\end{equation}
The values of $d_c$, $d_s$, and $d_{\mathrm{comm}}$ are chosen to enforce single-file traversal at each choke-point and convoy-style grouping between them.

As seen in Figures \ref{fig:hallway} and \ref{fig:hallway_perm}, the players display an emergent behavior of handling the order in which the players pass through the choke-points. This operation is then all done while maintaining the communication and collision constraints. Looking at the different permutations of the initial conditions and goals in Figure \ref{fig:hallway_perm}, the general behavior is somewhat uniform in subfigures a, b, c, and f where the agents enter the first chokepoint in sequence of distance away and then travel closely together. However, in subfigures d and e, one of the agents takes a slightly longer path instead of ``waiting" near the second chokepoint. Another aspect that is clear in Figure \ref{fig:hallway}, is the interplayer distance constraints being maintained in continuous time. Note how, even in between time steps, the constraints are upheld: albeit narrowly in some cases. 

\subsection{A Spacecraft Lady-Bandit-Guard Game}

\begin{figure}
    \centering
    \includegraphics[width=0.85\linewidth]{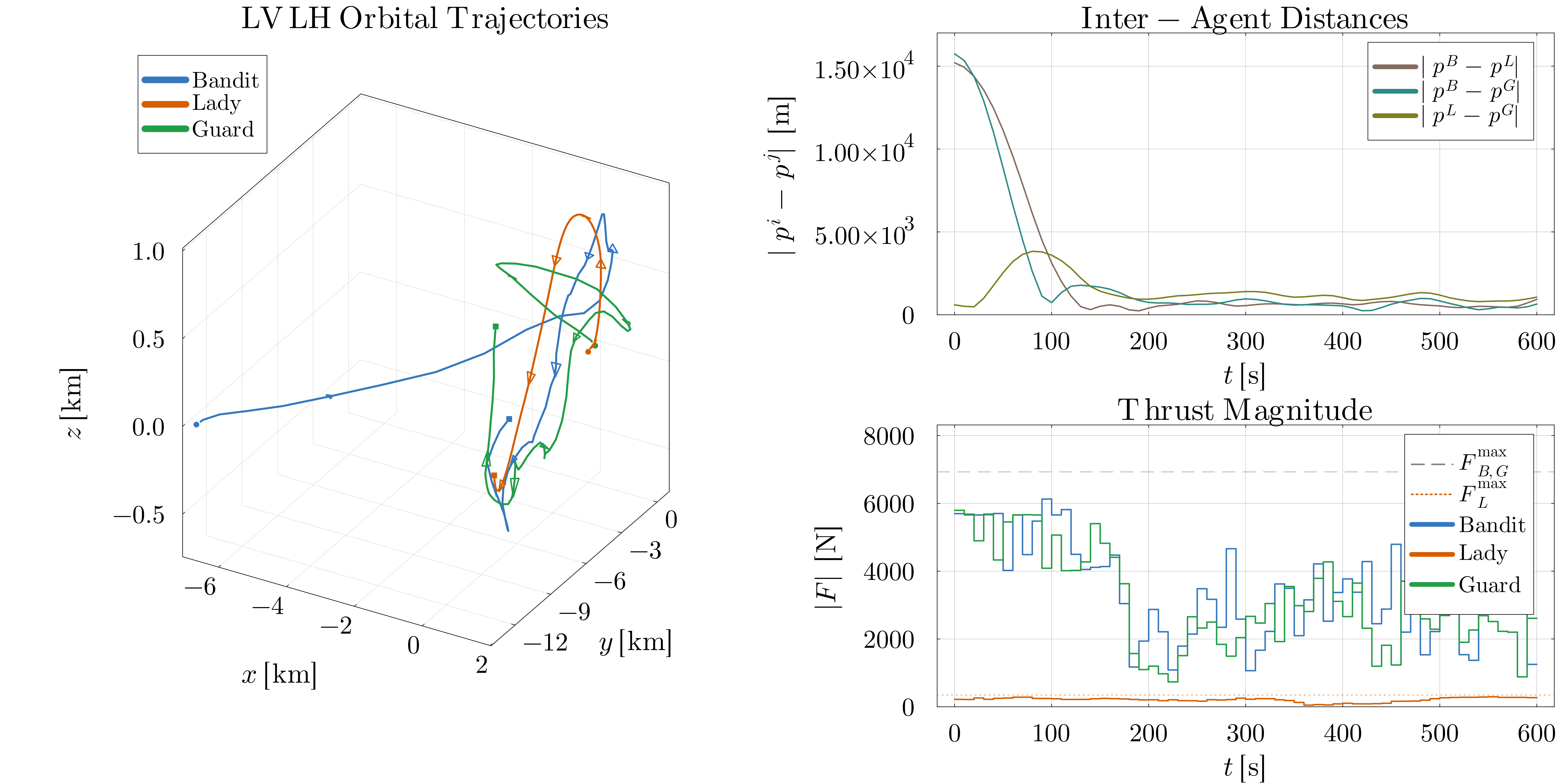}
    \caption{A three-player game planned for each player in $\mathbb{R}^6$ that consists of a lady, high value space asset, a guard spacecraft, and a bandit, attacker spacecraft. The bandit spacecraft seeks to enter a local region about the slowly fleeing lady spacecraft while the guard spacecraft attempts to block the attacker. The circles denote the initial conditions for the players while the squares denote the terminal states.}
    \label{fig:lbg}
\end{figure}

\begin{figure}
    \centering
    \includegraphics[width=\linewidth]{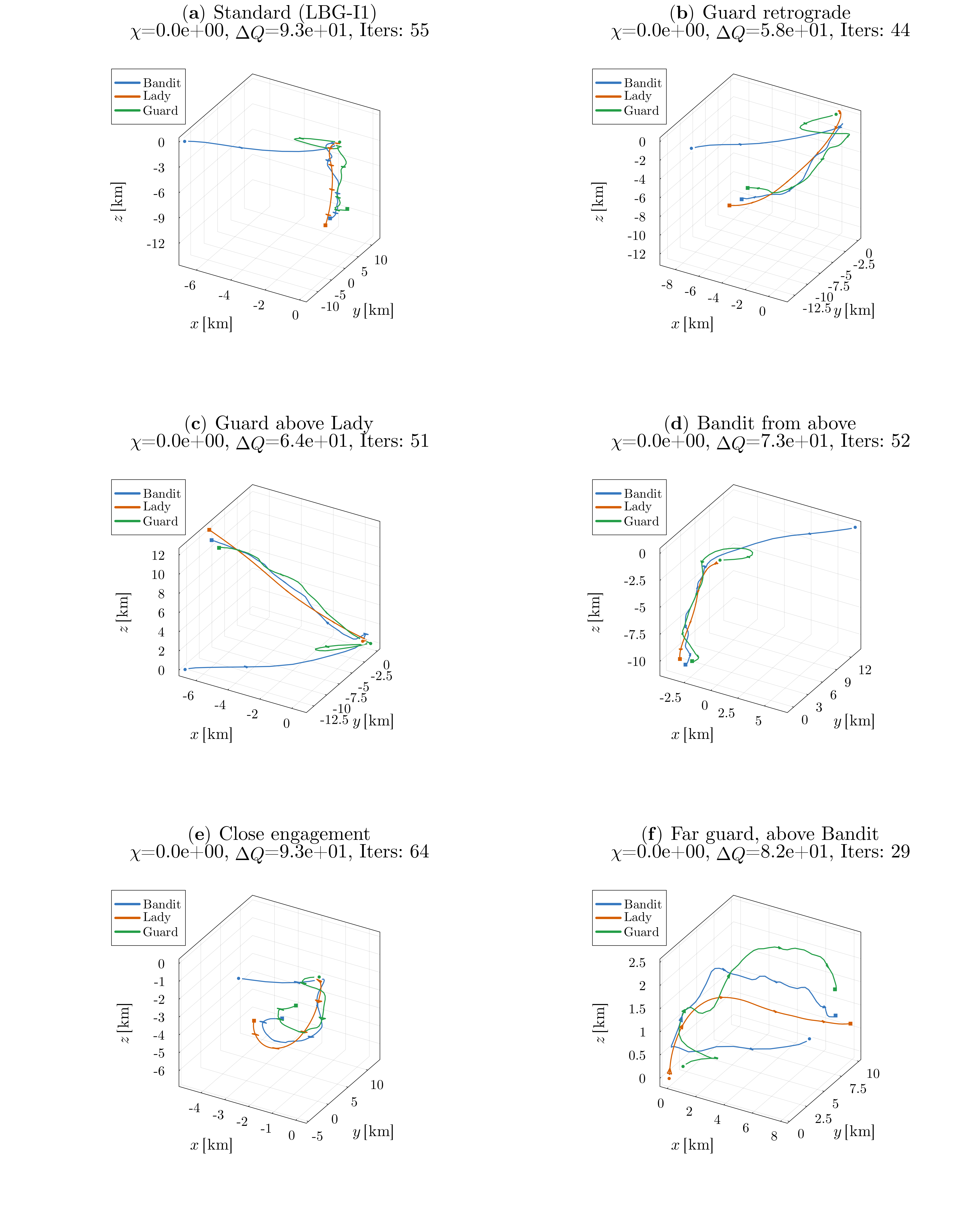}
    \caption{Permutations of the initial conditions for the guard and bandit in the Lady-Bandit-Guard game. Sub-figure a corresponds to the simulation in Figure \ref{fig:lbg}. The convergence certificate values are given at the top of each sub-figure. Note, the infeasibility certificate value is correct since all constraints, as formulated are convex. Throughout, there is a mixture of successful and unsuccessful games for the guard wherein the bandit achieved a commanding position of the lady.}
    \label{fig:lbg_perm}
\end{figure}

Consider a heterogeneous three-player game wherein each player operates under the linearized Clohessy-Wiltshire-Hill equations \cite{schaub2003analytical}. There is a slow moving high value space asset referred to as a ``lady" spacecraft. The second spacecraft is a ``bandit" spacecraft that seeks to capture the lady. Finally, a ``guard" spacecraft that protects the ``lady" spacecraft. This class of problem is sometimes referred to as a Target-Attacker-Defender Game that has been explored for spacecraft in elliptical orbits \cite{Target-Attacker-Defender}. However, we consider a case in close proximity that is further generalized. All spacecraft maintain appropriate collision constraints with respect to the other players. This is a problem that was inspired by the SpaceGym benchmark suite that now includes limited control authority for the lady and add collision constraints \cite{10115968}. More formally, let $d_k^{ij} = \|\mathbf{r}_k^i - \mathbf{r}_k^j\|_2$ denote the Euclidean separation between players $i$ and $j$ at time step $k$. Additionally, for normalization, let there be a reference distance, $d_{\mathrm{ref}}$. The objective of the bandit is to capture the lady while avoiding the guard,
\begin{equation}
    J^\mathcal{B}(z) =
        \left(\frac{d_T^{\mathcal{B}\mathcal{L}}}{d_{\mathrm{ref}}}\right)^{\!2}
        + \sum_{k=0}^{T-1}\left[
            \frac{w_{\mathcal{BG}}}{\dfrac{d_k^{\mathcal{B}\mathcal{G}}}{d_{\mathrm{ref}}} + 1}
            + w_{\mathrm{int}} \left(\frac{d_k^{\mathcal{B}\mathcal{L}}}{d_{\mathrm{ref}}}\right)^{\!2}
            + \frac{R_\mathcal{B}}{(F_{\max}^\mathcal{B})^2}\|u_k^\mathcal{B}\|_2^2
        \right],
\end{equation}
where $w_{\mathcal{BG}} > 0$ is the Guard-avoidance weight, $w_{\mathrm{int}} > 0$ is the cost for the approach, and the maximum thrust bound is $F_{\max}$. In opposition, the lady seeks to maximize separation from the bandit, 
\begin{equation}
    J^\mathcal{L}(z) =
        w_{\mathrm{evade}}\left(D_{\mathrm{ref}}^2 - \left(\frac{d_T^{\mathcal{B}\mathcal{L}}}{d_{\mathrm{ref}}}\right)^2\right)
        + \sum_{k=0}^{T-1}\left[
            w_{\mathrm{evade}}\left(D_{\mathrm{ref}}^2 - \left(\frac{d_k^{\mathcal{B}\mathcal{L}}}{d_{\mathrm{ref}}}\right)^2\right)
            + \frac{R_\mathcal{L}}{(F_{\max}^\mathcal{L})^2}\|u_k^\mathcal{L}\|_2^2
        \right],
\end{equation}
where $w_{\mathrm{evade}} > 0$. Finally, the guard seeks to intercept the bandit in defense of the lady, 
\begin{equation}
    J^\mathcal{G}(z) =
        \sum_{k=0}^{T-1}\left[
            w_{\mathrm{int}}\left(\frac{d_k^{\mathcal{B}\mathcal{G}}}{d_{\mathrm{ref}}}\right)^{\!2}
            + \frac{R_\mathcal{G}}{(F_{\max}^\mathcal{G})^2}\|u_k^\mathcal{G}\|_2^2
        \right],
\end{equation}
where $w_{\mathrm{int}} > 0$ is the interception cost weight and $R_i > 0$ is the control-cost weight for player $i \in \{\mathcal{B}, \mathcal{L}, \mathcal{G}\}$. All players maintain box constraints on control as well as inter-player collision constraints modeled by a radius around the spacecraft,
\begin{equation}
    d_k^{ij} \geq d_g,
    \qquad \forall\, i < j \in [N],\; \forall\, k \in \{0,\dots,T\}.
\end{equation}

The numerical results of this game can be seen in detail in Figure \ref{fig:lbg} and for varied initial conditions in Figure \ref{fig:lbg_perm}. In Figure \ref{fig:lbg}, we see a successful scenario of the lady escaping from the bandit with the assistance of the guard blocking the bandit. Looking at the permutations in Figure \ref{fig:lbg_perm}, we see different interesting cases in which the lady takes different escape directions. Additionally, we continue to see repeated blocking attempts by the guard that are relatively thwarted in all subfigures except e. It is important to note that this not representative of the ``win probability" for the entire game.

\section{Conclusions} \label{sec: Conclusion}

In this work, we have presented a novel algorithm that applies sequential convexification to solve for Nash equilibria in nonconvex differential games named FALCON. To ensure safety, the state was augmented to enforce continuous-time constraints. This augmented problem was then convexified and solved as a potential game. Conditions for multiplier updates and trust regions sizes were then developed. We demonstrated that FALCON is guaranteed to globally converge to a Nash equilibrium which distinguishes this method from other differential game solvers and multi-agent reinforcement learning. This was numerically demonstrated for three examples. First, we considered a F1 racing example on a difficult corner of the Austin F1 track. For this game, we numerically demonstrated that FALCON is faster than and has better convergence rates than ALGAMES and DG-SQP. Additionally, we also numerically demonstrated the performance of the FALCON algorithm on a narrowing hallways problem to show a difficult nonconvex environment wherein players display emergent behavior. Finally, FALCON was also applied to a spacecraft lady-bandit-guard game.

\section*{Acknowledgment}
Outland would like to thank the Science, Math, and Research for Transformation (SMART) Scholarship for academic funding. Outland dedicates his contribution {\footnotesize \calligra S.D.G}.

\section*{Generative AI Disclosure}
Generative AI was used for the following in this work: literature review aid, software implementation of the algorithms in this work, and as an informal proof assistant and aid. All work has been independently verified by the authors.


\bibliography{ref}

\end{document}